\documentclass[11pt]{article}

\usepackage{mathpazo}
\usepackage{times}

\usepackage{amsthm}
\usepackage{amsmath}   
\usepackage{amsfonts}
\usepackage{amssymb}   
\usepackage{mathrsfs}  
\usepackage{mathtools}
\usepackage{cite}      

\usepackage{tikz}
\usepackage{pgfplots}
\usetikzlibrary{pgfplots.groupplots}
\usetikzlibrary{decorations.pathreplacing,decorations.markings}
\usepackage{color}
\usepackage[ruled,vlined,linesnumbered]{algorithm2e} 
\usepackage{enumerate}				
\usepackage{multicol}				

\usepackage{graphicx}				
\usepackage{calc}					

\usepackage{fullpage}
\usepackage[title,toc,titletoc,page]{appendix}

\usepackage{xfrac}
\usepackage{epstopdf}

\newcommand\numberthis{\addtocounter{equation}{1}\tag{\theequation}}

\usepackage{graphicx}  
\usepackage{calc}      

\newcommand{\ccc}{\textbf{c}}
\newcommand{\eee}{\textbf{e}}
\newcommand{\xxx}{\textbf{x}}

\newcommand{\uuu}{\textbf{u}}
\newcommand{\vvv}{\textbf{v}}

\newcommand{\F}{\mathbb{F}}

\newcommand{\supp}{\mbox{supp}}

\newcommand{\wt}{\text{wt}}

\newcommand{\mcC}{\mathcal{C}}

\newtheorem{theorem}{Theorem}
\newtheorem{lemma}{Lemma}
\newtheorem{assumption}{Assumption}

\theoremstyle{definition}
\newtheorem{definition}{Definition}

\newcommand{\la}{\langle}
\newcommand{\ra}{\rangle}

\renewcommand{\le}{\leqslant} 
\renewcommand{\leq}{\leqslant}
\renewcommand{\ge}{\geqslant}

\title{\Huge $\,$\\[-6.00ex]
Explicit Polar Codes with\\
Small Scaling Exponent
}

\author{
Hanwen Yao\\
   \small University of California San Diego\vspace*{-0.72ex}\\
   \small 9500 Gilman Drive, La Jolla, CA\,92093\vspace*{-0.54ex}\\
   \ttfamily\bfseries\small hwyao@ucsd.edu\\[4.5ex]
\and
Arman Fazeli\\
   \small University of California San Diego\vspace*{-0.72ex}\\
   \small 9500 Gilman Drive, La Jolla, CA\,92093\vspace*{-0.54ex}\\
   \ttfamily\bfseries\small afazelic@ucsd.edu\\[4.5ex]
\and
{Alexander Vardy}\\
   \small University of California San Diego\vspace*{-0.72ex}\\
   \small 9500 Gilman Drive, La Jolla, CA\,92093\vspace*{-0.54ex}\\
   \ttfamily\bfseries\small avardy@ucsd.edu\\[6.5ex]
}

\begin{document}
\maketitle
\begin{abstract}
Polar coding gives rise to the first explicit family of codes that
provably achieve capacity for a wide range of channels with efficient 
encoding and decoding. But how fast can polar coding approach capacity 
as a function of the code length?
In finite-length analysis, the scaling between code length 
and the gap to capacity is usually measured in terms of the 
\emph{scaling exponent} $\mu$. It is well known that the optimal
scaling exponent, achieved by random binary codes, is $\mu = 2$.
It is also well known that the scaling %
exponent of conventional polar codes on the binary erasure channel 
(BEC) is $\mu =3.627$, which falls far short of the optimal value. %
On the other hand, it was recently shown that polar codes derived %
from $\ell\times\ell$ \emph{binary polarization kernels} approach the 
optimal~scaling exponent $\mu = 2$ on the BEC as $\ell \to \infty$, 
with high probability over a random choice of the kernel.

Herein, we focus on \emph{explicit constructions} of $\ell\times\ell$ 
binary~kernels with small scaling exponent 
for $\ell \le 64$. 
In particular,~we~exhibit
a sequence of binary linear codes that 
approaches capacity on the BEC with quasi-linear complexity 
and scaling~exponent $\mu < 3$. To the best of our knowledge, such
a sequence of codes~was 
not previously known to exist. 
The principal challenges~in~establishing our results are twofold:
how to construct such kernels~and %
how to evaluate their scaling exponent. 

In a single polarization step, an $\ell\times\ell$ kernel $K_\ell$
transforms~an underlying BEC into $\ell$ bit-channels $W_1,W_2,\ldots,W_\ell$.
The erasure probabilities of $W_1,W_2,\ldots,W_\ell$, known as the
\emph{polarization behavior of $K_\ell$}, determine the resulting
scaling exponent $\mu(K_\ell)$. 
We first introduce a 
class of \emph{self-dual} binary kernels and prove~that~their %
polarization behavior satisfies a strong symmetry property.
This reduces the problem of constructing $K_\ell$ to that of 
producing~a~cer\-tain nested chain of only $\ell/2$ self-orthogonal codes.
We use nested %
cyclic codes, whose distance is as high as possible
subject to~the orthogonality constraint, to construct the kernels 
$K_{32}$ and $K_{64}$.
In order to evaluate the polarization behavior of $K_{32}$ and $K_{64}$,
two alternative trellis representations (which may be of independent
interest) are proposed. 
Using the resulting trellises, we show~that %
$\mu(K_{32})=3.122$ and explicitly compute over half of the 
polariza\-tion-behavior 
coefficients for $K_{64}$, at which point the complexity becomes
prohibitive. To complete the computation, 
we introduce a Monte-Carlo interpolation method, which produces
the estimate
$\mu(K_{64})\simeq 2.87$. We augment this estimate with a rigorous proof that
$\mu(K_{64})< 2.97$.
\end{abstract}

\section{Introduction} 
\label{sec:Introduction}

\noindent 
Polar coding, pioneered by Ar{\i}kan in~\cite{Arikan},
gives rise to the~first explicit family of codes that
provably achieve capacity for a wide range of channels with efficient 
encoding and decoding. This paper is concerned with 
\emph{how fast can polar coding~approach capacity}
as a function of the code length?
In finite-length analysis~\cite{%
rdkernel,
Hassani,
MHU16,
PU16,
PPV10}, 
the scaling between code length $n$
and the gap to capacity 
is usually measured in terms of the \emph{scaling exponent} $\mu$. 
It is well known that the scaling
exponent of conventional polar codes on the 
BEC is $3.627$,~which falls far short of the optimal value
$\mu = 2$.
However, 
it was recently shown~\cite{rdkernel}
that polar codes derived
from $\ell\times\ell$ {polarization~kernels} approach 
optimal scaling 
on the BEC as $\ell\,{\to}\,\infty$, 
with high probability over a random choice of the kernel.

Korada, \c{S}a\c{s}o\u{g}lu, and Urbanke~\cite{Korada} were the 
first to show that polarization theorems still hold if one replaces 
the conventional $2\times 2$ kernel $K_2$ of Ar\i kan~\cite{Arikan}
with an $\ell\times\ell$~binary matrix, provided that this
matrix is nonsingular and~not upper
triangular under any column permutation. 
Moreover,~\cite{Korada}~esta\-blishes a simple formula for
the error exponent of the resulting
polar codes in terms
of the {partial distances} of certain nested \emph{kernel codes}. 
However, an explicit formulation for the scaling
exponent is at present unknown, 
even for the simple case~of the BEC.
Just like Ar\i kan's 
$2\times 2$ kernel $K_2$, which transforms
the underlying channel
$W$ into two synthesized bit-channels
$\{W^+,W^-\}$, an $\ell\times\ell$ kernel
$K_\ell$ transforms $W$ into~$\ell$~synthesized 
{bit-channels} $W_1,W_2,\ldots,W_\ell$. If $W$ is a BEC 
with~erasure probability $z$, the bit-channels 
$W_1,W_2,\ldots,W_\ell$ are also BECs and their erasure probabilities
are given by integer polynomials $f_i(z)$ for $i = 1,2,\ldots,\ell$. The 
set $\{f_1(z),f_2(z),\ldots,f_\ell(z)\}$~is known~\cite{Allerton,rdkernel}
as the \emph{polarization behavior of $K_\ell$} 
and completely determines its scaling exponent $\mu(K_\ell)$.
 	
\begin{figure}[t]
 		\centering
 		\begin{tikzpicture}
 		
 		\begin{groupplot}[
 		group style={
 			group name=my fancy plots,
 			group size=1 by 2,
 			xticklabels at=edge bottom,
 			vertical sep=0pt,
 		},
 		xmin=0.5, xmax=6.5,
 		height = 7cm,
 		width = 12cm,
 		]
 		
 		\nextgroupplot[
 		axis y line=left,
 		axis x line=bottom,
 		x axis line style= {-,dashed,very thick},
 		xmin=0.5,xmax=6.5,
 		ymin=2.6,ymax=3.7,
 		axis y discontinuity=crunch,
 		ytick={2.6,2.8,3,3.2,3.4,3.6},yticklabels={2.0,2.8,3.0,3.2,3.4,3.6},
 		grid,
 		ylabel={scaling exponent $\mu$}]
 		
 		\addplot+[
 		nodes near coords,
 		point meta=explicit symbolic,
 		color=black,mark options={black},
 		visualization depends on={value \thisrow{anchor}\as\myanchor},,
 		every node near coord/.append style={font=\small,anchor=\myanchor}
 		] table [row sep=\\,meta=Label] 
 		{
 			x y Label anchor\\
 			1 3.627 {$3.627$} -90\\
 			2 3.627 {$3.627$} -90\\
 			3 3.577 {$3.577$} -120\\
 			4 3.346 {$3.346$} -140\\
 			5 3.122 {$3.122$} -140\\
 			6 2.87 {$2.87$} -140\\
 		};
 		\addplot+[dashed,very thick,color=black,mark=none]
 		coordinates {(0.5,3) (6.5,3)};
 		\node[] at (axis cs: 1.9,2.65) {optimal scaling exponent};
 		
 		\nextgroupplot[ymin=2.4,ymax=2.6,yscale=0.09,grid,
 		xtick={1,2,3,4,5,6},xticklabels={2,4,8,16,32,64},
 		ytick={2.4},yticklabels={},
 		axis x line=bottom,
 		axis y line=left,
 		y axis line style= {-},
 		xlabel={kernel size $\ell$},
 		]           
 		\end{groupplot}
 		\end{tikzpicture}
\caption{Scaling exponents of binary polarization kernels of size $\ell$. 
The values for $\ell=2,4,8$ are optimal; the values for $\ell=16,32,64$ 
are best known.}
 		\label{fig:scaling_curves}
\end{figure}
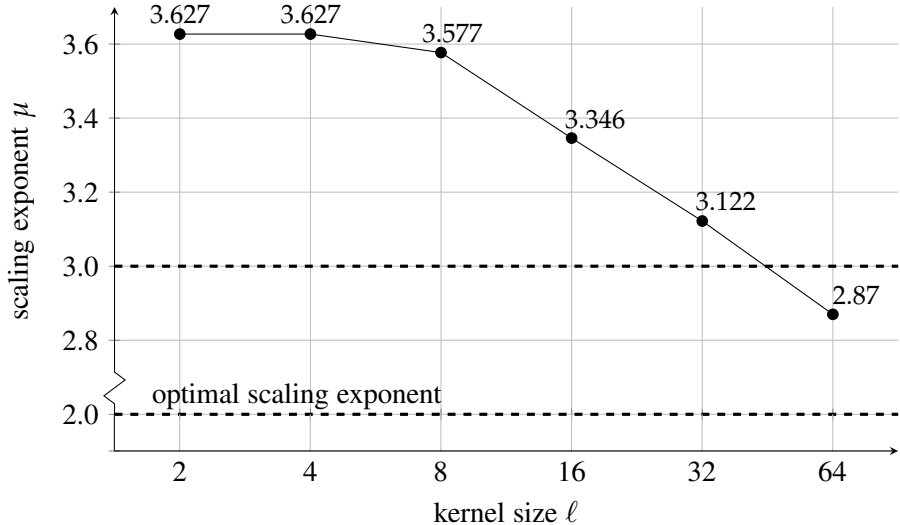
 	
While smaller scaling exponents 
translate into better finite-length performance, 
the complexity of decoding 
can grow~exponentially with the kernel size. There
have been attempts~to reduce the decoding complexity of large
kernels~\cite{BFSTV17,Trifonov},~however this
problem remains unsolved in general.
We note that, although our constructions are explicit, 
issues such as decoding the kernel are beyond the scope 
of this work. Rather, our goal is to study the following 
simple question: what is the smallest scaling exponent one can get with
an $\ell\times\ell$ binary kernel?
In particular, we construct a kernel $K_{64}$ with $\mu(K_{64}) \simeq 2.87$. 
This
gives rise to 
a sequence of binary linear codes that 
approaches capacity on the BEC with quasilinear complexity 
and scaling exponent strictly less than $3$.
To the best of our~knowledge such a sequence of codes was
not previously known to exist.

 	
\subsection{Related Prior Work}
\noindent 
Scaling exponents of error-correcting codes have been subject to
an extensive amount of research. It was known since the work of
Strassen~\cite{S62} that random codes attain the optimal scaling 
exponent $\mu =2$. It was furthermore shown 
in~\cite{PPV10} that random \emph{linear} codes
also achieve this optimal value. For polar codes, the first
attempts at bounding their scaling exponents were given
in~\cite{Hassani}, where the scaling exponent of polar codes for
arbitrary channels were shown to be bounded by $3.579\leq \mu \leq
6$. The upper bound was improved to $\mu \leq 4.714$
in~\cite{MHU16}. An upper bound on the scaling exponent of polar codes
for non-stationary channels was also presented in~\cite{M16} as $\mu
\leq 10.78$.
 	
Authors in~\cite{Hassani} also introduced a method to explicitly calculate the 
scaling exponent of polar codes over BEC based on its polarization behavior. They 
showed that for the Ar\i kan's kernel $K_2$, $\mu = 3.627$. Later on, an 
$8\times 8$ kernel $K_8$ was found with $\mu=3.577$ for BEC, which is optimal among 
all kernels with $\ell\le 8$~\cite{Allerton}. It was accompanied with a heuristic 
construction to design larger polarizing kernels with smaller scaling exponents, which 
gave rise to a $16\times 16$ kernel with $\mu=3.356$. In \cite{Trifonov}, a 
$32\times 32$ kernel $F_{32}$ and a $64\times 64$ kernel was constructed, which was 
shown (via simulations) to have a better \emph{frame error rate} than the Ar\i kan's 
kernel. They have also introduced an algorithm based on the 
\emph{binary decision diagram} (BDD) to efficiently calculate the polarization 
behavior of larger kernels. Attempts to achieve the optimal scaling exponent of $2$ 
were first seen in~\cite{PU16}, where it was shown that polar codes can achieve the 
near-optimal scaling exponent of $\mu =2+\epsilon$ by using explicit large kernels 
over large alphabets. The conjecture was just recently solved in~\cite{rdkernel}, 
where it was shown that one can achieve the near-optimal scaling exponent via 
\emph{almost any binary} $\ell\times\ell$ kernel given that $\ell$ is sufficiently 
large enough. Now it remains to find the explicit constructions of such 
\emph{optimal} kernels. Our results in this paper can be viewed as another step 
towards the derandomization of the proof in~\cite{rdkernel}.
 	

\subsection{Our Contributions}

\noindent 
In this paper, a more comprehensive kernel construction approach is proposed. We first 
introduce a special class of polarizing kernels called the \emph{self-dual} kernels. 
For those self-dual kernels, we prove a duality theorem showing that their 
polarization behaviors are symmetric, which enables us to construct the kernel by 
only designing its bottom half. In our construction, we use a greedy approach for 
the bottom half of the kernel, where we push the values of $f_i(z)$ as close to $0$ 
as possible in the order of $i=\ell, \ell-1, \cdots$, which intuitively gives us 
small scaling exponents. This construction gives the best previously found 
$16\times 16$ kernel $K_{16}$ provided in \cite{TT18} with scaling exponent $3.346$, 
a new $32\times 32$ kernel $K_{32}$ with $\mu(K_{32}) = 3.122$, and a new 
$64\times 64$ kernel $K_{64}$ with $\mu(K_{64}) \simeq 2.87$ as depicted in 
Figure~\ref{fig:scaling_curves}. We utilize the partial distances of nested 
Reed-Muller (RM) codes and cyclic codes to implement the proposed construction 
approach. 
 	
To calculate the scaling exponent of our constructed kernels, we first calculate their 
polarization behaviors, and then invoke the method introduced in \cite{Hassani}. 
For a specific bit-channel, its polarization behavior polynomial $f_i(z)$ can be 
described by the weight distribution of its \emph{uncorrectable erasure patterns}. 
To calculate this weight distribution, we introduce a new trellis-based algorithm. 
Our algorithm is significantly faster than the BDD based algorithm proposed in 
\cite{Trifonov}. It first builds a \emph{proper trellis} for those uncorrectable 
erasure patterns, and then applies the Viterbi algorithm to calculate its weight 
distribution. We also propose an alternative approach that builds a 
\emph{stitching trellis}, which we believe is of independent interest. However, for 
a very large kernel ($K_{64}$ in our case), the complexity of our trellis algorithm 
gets prohibitively high for intermediate bit-channels. As a fix, we introduce an 
alternative Monte Carlo interpolation-based method to numerically estimate those 
polynomials of the intermediate bit-channels, which we use to estimate the scaling 
exponent of $K_{64}$ as $\mu(K_{64}) \simeq 2.87$. We further give a rigorous
proof that $\mu(K_{64})< 2.97$.
 	

\section{Preliminary Discussions}
\noindent
Let $K_\ell$ be a $\ell\times\ell$ kernel $K_\ell=[g_1^T,g_2^T,\cdots,g_\ell^T]^T$ and 
$\xxx=\uuu K_\ell$ be a codeword that is transmitted over $\ell$ i.i.d. BEC channels 
$W=\text{BEC}(z)$. We define an \emph{erasure pattern} to be a vector 
$\eee\in\{0,1\}^\ell$, where 1 corresponds to the erased positions of $\xxx$ and 0 
corresponds to the unerased positions. The probability of occurance of a specific 
erasure pattern $\eee$ will be $z^{\wt(\eee)}(1-z)^{\ell-\wt(\eee)}$, where 
$\wt(\eee)$ is the Hamming weight of $\eee$.

\begin{definition}[Uncorrectable Erasure Patterns]
We say the erasure pattern $\eee$ is \emph{uncorrectable} for a bit-channel $W_i$ 
if and only if there exists two information vectors $\uuu',\uuu''$ such that 
$u'_j=u''_j$ for $j<i$, $u'_i\neq u''_i$ and $(\uuu'K_\ell)_j=(\uuu''K_\ell)_j$ for 
all unerased positions $j\in\{k\,:\,e_k=0\}$.
\label{def_e}
\end{definition}

\noindent
For the $i$-th bit-channel $W_i$, let $E_{i,w}$ be the number of its uncorrectable erasure patterns of weight $w$, then its erasure probability $f_i(z)$ can be represented as the polynomial

\begin{equation}
f_i(z)=\sum_{w=0}^{\ell}E_{i,w}z^w(1-z)^{(\ell-w)}.
\label{fiz}
\end{equation}

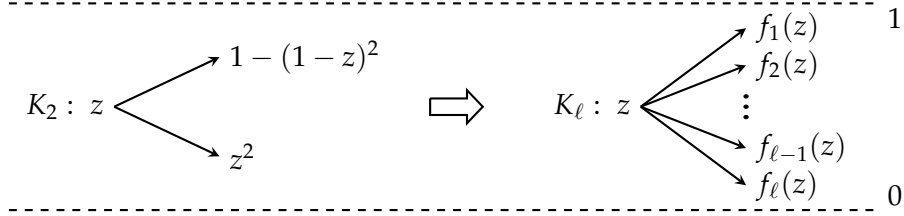
\begin{figure}[t]
 		\centering
 		\begin{tikzpicture}[xscale=0.7,yscale = 0.55]
 		\tikzset{>=stealth}
 		
 		\begin{scope}[shift={(-10,0)}]
 		\node [left] at (-1,0) {$K_2:\; z$};
 		\def\a{1.2};
 		\draw [thick,->] (-1,0) -- (1,\a);
 		\draw [thick,->] (-1,0) -- (1,-\a);
 		\node [right] at (1,\a) {$1-(1-z)^2$};
 		\node [right] at (1,-\a) {$z^2$};
 		\end{scope}
 		\begin{scope}[shift={(-9.5,-2)}]
 		\draw [thick] (4.5,2.2) -- (5.2,2.2) -- (5.2,2.4) -- (5.5,2) -- (5.2,1.6)
 		-- (5.2,1.8) -- (4.5,1.8) -- (4.5,2.2);
 		\end{scope}
 		\node [left] at (-1,0) {$K_\ell:\;z$};
 		
 		\def\a{1.9};
 		\draw [thick,->] (-1,0) -- (1,\a);
 		\draw [thick,->] (-1,0) -- (1,1);
 		\draw [thick,->] (-1,0) -- (1,-\a);
 		\draw [thick,->] (-1,0) -- (1,-1);
 		\node [right] at (1,\a) {$f_1(z)$};
 		\node [right] at (1,1) {$f_{2}(z)$};
 		\node [right] at (1,-\a) {$f_{\ell}(z)$};
 		\node [right] at (1,-1) {$f_{\ell-1}(z)$};
 		
 		\foreach \i in {-0.25,0,0.25}
 		{\draw [fill] (1,0+\i) circle [radius=0.04];}
 		
 		\def\b{2.5};
 		\def\c{3.5};
 		\draw [thick, dashed] (-13,\b) -- (\c,\b);
 		\draw [thick, dashed] (-13,-\b) -- (\c,-\b);
 		\node [right] at  (\c,\b-0.35) {$1$};
 		\node [right] at  (\c,-\b+0.35) {$0$};
 		\end{tikzpicture}
 		\caption{Transformation of the erasure probabilities in one polarization step.}
 		\label{fig_polarize}
\end{figure}
 	
\noindent 
Therefore if we can calculate the weight distribution of its uncorrectable erasure 
patterns $E_{i,0},E_{i,1},\cdots, E_{i,\ell}$, we can get the polynomial $f_i(z)$. We 
call the entire set $\{f_1(z),\cdots,f_\ell(z)\}$ as the \emph{polarization behavior} 
of $K_\ell$. One can utilize the techniques in~\cite{Hassani} to estimate the scaling 
exponent of polar codes with large kernels by replacing the transformation polynomials 
$\{z^2, 1-(1-z)^2\}$ in the traditional polar codes with the polarization behavior of 
$K_\ell$ defined above.
 	
\section{Construction of Large Self-dual Kernels}
 	
\subsection{Kernel Codes}

\noindent 
Before we find out what those uncorrectable erasure patterns are, we give the 
following definitions. Given two vectors $\vvv_1,\vvv_2$, we say $\vvv_2$ 
\emph{covers} $\vvv_1$ if $\supp(\vvv_1)\subseteq\supp(\vvv_2)$. Given a set 
$\mathcal{S}\subseteq \F_2^\ell$, we define its \emph{cover set} 
$\Delta(\mathcal{S})$ as the set of vectors that covers at least one vector in 
$\mathcal{S}$. It will be shown later, that the set of those uncorrectable erasure 
patterns are the cover set of a coset.

\begin{definition}[Kernel Codes]
Given an $\ell\times\ell$ kernel $K_\ell$, we define the \emph{kernel codes} 
$\mcC_i=\la g_{i+1},\cdots,g_\ell\ra$ for $i=0,\cdots,\ell$, and $\mcC_\ell=\{0\}$.
\end{definition}

\begin{theorem}
An erasure pattern $\eee$ is uncorrectable for $W_i$ if and only if 
$\eee\in\Delta(\mcC_{i-1}\backslash\mcC_i)$.
\label{delta}
\end{theorem}
\begin{proof}
Here, we prove the ``only if'' direction. The other direction follows similarly. If 
$\eee$ is uncorrectable, then there exists $\uuu',\uuu''$ as described in Definition 
\ref{def_e}. So $(\uuu'-\uuu'')_j=0$ for $j<i$ and $(\uuu'-\uuu'')_i=1$. Thus 
$\ccc=(\uuu'-\uuu'')K_\ell$ is a codeword in the coset $(\mcC_{i-1}\backslash\mcC_i)$. 
Also, since $\uuu'K_\ell$ and $\uuu'K_\ell$ agree on all the unerased positions, this 
codeword $\ccc=\uuu'K_\ell-\uuu''K_\ell$ is covered by the erasure pattern $\eee$. 
So $\eee\in\Delta(\mcC_{i-1}\backslash\mcC_i)$.
\end{proof}
 	
\subsection{Self-dual Kernels and Duality Theorem}

\noindent 
We first introduce a special type of \emph{self-dual} kernels. We call an 
$\ell\times\ell$ kernel \emph{self-dual} if $\mcC_i=\mcC_{\ell-i}^\perp$ for all 
$i=0,\cdots,\ell$. Then we prove the duality theorem, which shows that the 
polarization behavior of a self-dual kernel is symmetric.

\begin{lemma}
 		If $K_\ell$ is self-dual, then
 		\begin{align}
 		\forall_{i}\ \forall_{w}:\quad E_{i,w}+E_{\ell+1-i,\ell-w}\le \binom{\ell}{w}
 		\label{dualityequation}
 		\end{align}
 		\label{dualitylemma}
\end{lemma}
\begin{proof}
Let $\eee$ be an uncorretable erasure pattern for $W_i$. Assume $\eee$ is 
uncorrectable for $W_i$ while its complement $\overline{\eee}$ is also uncorrectable 
for $W_{\ell+1-i}$, then $\eee$ covers a codeword $\ccc_1$ in 
$(\mcC_{i-1}\backslash\mcC_i)$ and $\overline{\eee}$ covers a codeword $\ccc_2$ in 
$(\mcC_{\ell-i}\backslash\mcC_{\ell+1-i})=(\mcC_{i-1}^\perp\backslash\mcC_{i}^\perp)$. 
Since $\supp(\eee)$ and $\supp(\overline{\eee})$ are disjoint, we have 
$\ccc_1\perp \ccc_2$. But since $\mcC_{i-1}$ only has one more dimension than 
$\mcC_i$, $\ccc_2\perp \mcC_i$ and $\ccc_2\perp \vvv$ would imply 
$\ccc_2\perp\mcC_{i-1}$, which is a contradiction. Therefore the complement 
$\overline{\eee}$ of every uncorrectable erasure pattern $\eee$ for  $W_i$ is 
correctable for $W_{\ell+1-i}$, which yields in the proof.
\end{proof}

\begin{theorem}[Duality Theorem]
		If $K_\ell$ is self-dual, then for $i=1,\cdots,\ell$
 		\begin{equation}
 		f_{\ell+1-i}(z)=1-f_i(1-z)
 		\end{equation}
 		\label{duality}
\end{theorem}
\begin{proof}
 		For all $i=1,\cdots,\ell$ we have
 		\begin{align*}
 		\looseness=-1
 		f_i&(z)+f_{\ell+1-i}(1-z)\\
 		&=\sum_{w=0}^{\ell}(E_{i,w}+E_{\ell+1-i,\ell-w})z^w(1-z)^{\ell-w}\hspace{2mm}\stackrel{(\ref{dualityequation})}{\leq} 1.
 		\numberthis\label{fiz+fi1-z}
 		\end{align*}
 		Therefore, $\sum_{i=1}^{\ell}(f_i(z)+f_{\ell+1-i}(1-z))\leq \ell$. But a polarization step is a capacity preserving transformation, which means
 		\begin{equation}
 		\sum_{i=1}^{\ell}f_i(z)+\sum_{i=1}^{\ell}f_{\ell+1-i}(1-z)=\ell z+\ell(1-z)=\ell.
 		\end{equation}
 		So all the previous inequalities must hold with equality.
\end{proof}
 	
\subsection{Kernel Construction}

\begin{table}[t!]
 		\centering
 		\begin{tabular}{|l|l|c|}
 			\hline
 			rows & kernel codes & partial distances\\\hline
 			
 			32 & $\mcC_{32}=\{\bf{0,1}\}$ & 32\\\hline
 			
 			28-31 & subcodes of $\mcC_{27}$ & 16\\
 			27 & $\mcC_{27}=$ RM(1,5) & 16\\\hline
 			
 			23-26 & subcodes of $\mcC_{12}$ & 12\\
 			22 & $\mcC_{12}=$ extended BCH(31,11,11) & 12\\\hline
 			
 			18-21 & subcodes of $\mcC_{17}$ & 8\\
 			
 			17 & $\mcC_{17}=$ RM(2,5) & 8\\\hline
 		\end{tabular}
 		\caption{Kernel codes of $K_{32}$ at the bottom half}
 		\label{table32}
 	\end{table}
 	\begin{table}[t!]
 		\centering
 		\begin{tabular}{|l|l|c|}
 			\hline
 			rows & kernel codes & partial distances\\\hline
 			
 			64 & $\mcC_{64}=\{\bf{0,1}\}$ & 64\\\hline
 			
 			59-33 & subcodes of $\mcC_{58}$ & 32\\
 			58 & $\mcC_{58}=$ RM(1,6) & 32\\\hline
 			
 			56-57 & subcodes of $\mcC_{55}$ & 28\\
 			55 & $\mcC_{55}=$  extended BCH(63,10,27) & 28\\\hline
 			
 			50-54 & subcodes of $\mcC_{49}$ & 24\\
 			49 & $\mcC_{49}=$  extended BCH(63,16,23)& 24\\\hline
 			
 			44-48 & subcodes of $\mcC_{43}$ & 16\\
 			43 & $\mcC_{43}=$  RM(2,6) & 16\\\hline
 			
 			38-42 & subcodes of $\mcC_{37}$ & 16\\
 			37 & $\mcC_{37}=$  extended cyclic(63,28,15)& 16\\\hline
 			
 			36 & $\mcC_{36}=$ (64,29,14) linear code & 14\\
 			35 & $\mcC_{35}=$ (64,30,12) linear code & 12\\
 			34 & $\mcC_{34}=$ (64,31,12) linear code & 12\\
 			33 & $\mcC_{33}=$ (64,32,12) linear code & 12\\\hline
 		\end{tabular}
 		\caption{Kernel codes of $K_{64}$ at the bottom half}
 		\label{table64}
\end{table}

\noindent 
The intuition behind our kernel construction is to \textbf{a)} mimic the polarization 
behavior of random kernels by making $f_i(z)$'s jump from $f_i(z)\in (0,\epsilon)$ to 
$(1-\epsilon,1)$ as sharp as possible (see Figure~\ref{fig_B32}). \textbf{b)} provide 
a symmetry property in which half of the polynomials are polarizing to the value of 
$0$ and the other half are polarizing to the value of $1$ as depicted in 
Figure~\ref{fig_polarize}. In each step of our construction algorithm, we make sure 
that the constructed kernel is self-dual to design a symmetric polarization behavior 
according to the the duality theorem. This allows us to focus on constructing only one 
half of the kernel. Here, we pick the bottom half. The strategy behind constructing 
the bottom half is to construct the rows in kernel one by one, while maximizing the 
\emph{partial distance}, defined below, in each step.

\begin{definition}[Partial Distances]
Given an $\ell\times\ell$ kernel $K_\ell$, we define the \emph{partial distances} 
$d_i= d_H(g_i,\mcC_i)$ for $i=1,\cdots,\ell-1$, and $d_\ell= d_H(g_\ell,0)$.
\end{definition}

\noindent 
When $z$ is close to 0, the polynomial $f_i(z)$ will be dominated by the first 
non-zero term $E_{i,w}z^w(1-z)^{(\ell-w)}$. By Theorem \ref{delta} the first non-zero 
coefficients of $f_i(z)$ is $E_{i,d_i}$. So, we aim to maximize the partial distance 
$d_i$ to make $f_i(z)$ polarize towards 0.
 	
The construction algorithm in a nutshell is described in the following. Start by 
setting $\mcC_\ell=\{0\}$. Then from the bottom upwards, construct the bottom half of 
the kernel row by row greedily with maximum possible partial distances, while 
maintaining the kernel's self-dual property. Namely for $i$ from $\ell$ to $\ell/2+1$, 
pick $v\in(\mcC_i^\perp\backslash\mcC_i)$ with the maximum partial distance 
$d_i=d(v,\mcC_i)$ to be the $i$-th row of the kernel. The construction of the other 
half follows immediately by preserving the self-duality in each step. 
 	
Let us implement the algorithm for $\ell=32$. We first pick the bottom row $g_{32}$ of 
$K_{32}$ to be the all 1 vector \textbf{1}. Then for row 27-31, we pick codewords in 
RM(1,5) with maximum partial distance 16. After that, we carefully select codewords in 
the extended BCH codes and the RM(2,5), that both have maximum partial distances, and 
preserve the self-dual property of the kernel. The kernel code $\mcC_{17}$ happens to 
be exactly the self-dual code RM(2,5). We finish our construction by filling up the 
top half and get the self-dual kernel $K_{32}$ as shown in Fig~\ref{fig:K32}. We 
construct $K_{64}$ as shown in Fig~\ref{fig:K64} similarly, except that row 33-36 are 
picked through computer search. The kernel codes at the bottom half of $K_{32}$ and 
$K_{64}$ are shown in Table~\ref{table32},\ref{table64}.
 	
\begin{figure}[t!]
 		\centering
 		\includegraphics[height=7.5cm]{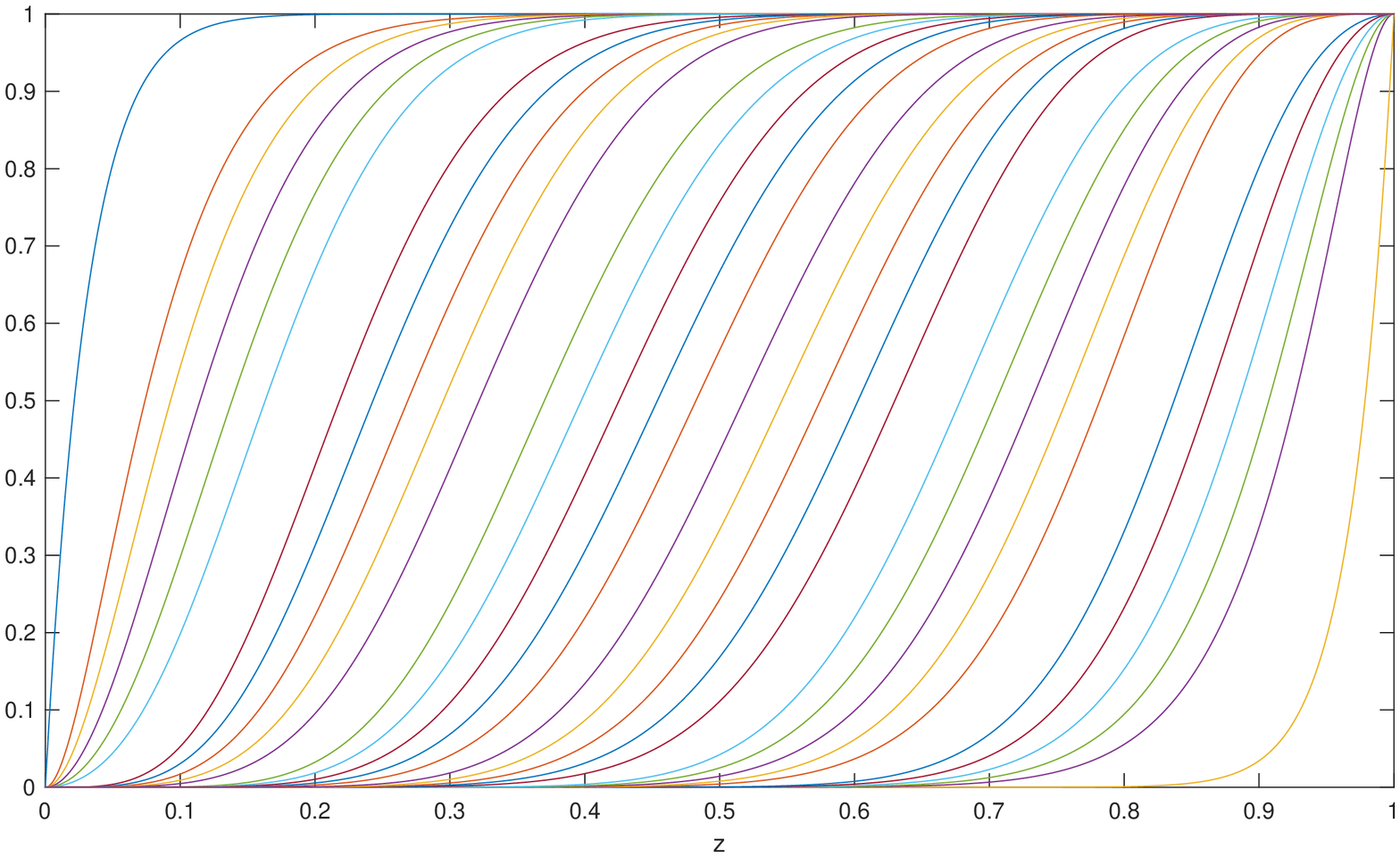}
 		\caption{Polarization behavior of kernel $K_{32}$}
 		\label{fig_B32}
\end{figure}


\section{Calculate the Polarization Behaviors}
\noindent 
So far, we presented an algorithm to construct large binary kernels with intuitively 
good scaling exponents. In this section, we address the last challenge, which is to 
efficiently derive the polarization behavior of a given kernel. The NP hardness of 
this problem was previously established in~\cite{Allerton}. In this section, we 
propose a few methods to reduce the computation complexity just enough so we can 
implement it. To this end, we present two trellis-based algorithms that can explicitly 
calculate the polarization behavior of $K_{32}$. Sadly, even these improved algorithms 
are beyond implementation for $K_{64}$. So, we present an alternative approach of 
``estimating'' the polarization behavior of $K_{64}$ with high precision using a large 
but limited number of samples from the set of all erasure patterns. One can plug the 
estimated polarization behavior into the methods described in~\cite{Hassani} and get 
$\mu(K_{64}) \simeq 2.87$. We also provide a more careful analysis to show that 
$\mu(K_{64}) \leq 2.9603$ rigorously. 
 	
\subsection{Trellis Algorithms}
\noindent 
A trellis is a graphical representation of a block code, in which every path 
represents a codeword. This representation allows us to do ML decoding with reduced 
complexity using the famous Viterbi algorithm. The Viterbi algorithm allows one to 
find the most likely path in a trellis. Besides decoding, it can also be generalized 
to find the weight distribution of the block code, given that the trellis is 
\emph{one-to-one}. A trellis is called \emph{one-to-one} if all of its paths are 
labeled distinctly. We refer the readers to~\cite{trellisbook} for the known facts 
about trellises we use in this section.
 	
In this work, we develop new theory for trellis representation for the cover sets, 
which are both nonlinear and non-rectangular. We introduced two different algorithms 
that both can construct a one-to-one trellis for the cover set 
$\Delta(\mcC_{i-1}\backslash\mcC_i)$. By efficiently representing the cover sets using 
trellises, we can use the Viterbi algorithm to calculate its weight distribution. A 
brief description of these algorithms together are given in the following. An example 
is also provided in Figure~5 for interested readers to track the steps in both 
algorithms.
 	
 	{
 		\setlength{\textfloatsep}{-0pt}
 		\SetInd{0.5em}{0.5em}
 		\begin{algorithm}[t]
 			\caption{Construct a proper trellis $T^\ast$ from $T$}
 			\For{$i=0$ to $(\ell-1)$}
 			{
 				\For{\text{every vertex }$v^\ast_i\in V^\ast_i$}
 				{
 					\For{$a\in \{0,1\}$}
 					{
 						calculate $s=$\\
 						$\{v_{i+1}\in V_{i+1}:\exists v_i\in L(v^\ast_i),(v_i,v_{i+1},a)\in E\}$\\
 						\eIf{$\exists v^\ast_{i+1}\in V^\ast_{i+1}$ with $L(v^\ast_{i+1})=s$}
 						{
 							add an edge $(v^\ast_i,v^\ast_{i+1},a)$ in $E^\ast$
 						}
 						{
 							add a vertex $v^\ast_{i+1}\in V^\ast_{i+1}$ with $L(v^\ast_{i+1})=s$\\
 							add an edge $(v^\ast_i,v^\ast_{i+1},a)$ in $E^\ast$
 						}
 					}
 				}
 			}
 			\label{algo1}
 		\end{algorithm}
 	}
 	
\vspace{3mm}
\noindent 
\textbf{\emph{Proper Trellis Algorithm}}
\vspace{3mm}
 	
\noindent 
A trellis is called \emph{proper} if edges beginning at any given vertex are labeled 
distinctly. It is known that if a trellis is proper, then it is one-to-one. So, one 
way of constructing a one-to-one trellis for $\Delta(\mcC_{i-1}\backslash\mcC_i)$ is 
to construct a proper trellis. The proper trellis algorithm has the following steps. 
\textbf{Step 1:} Construct a minimal trellis for the linear code $\mcC_i$. For every 
edges in $E_i$ where $i\in\supp(g_i)$, flip its label. We can then get a trellis for 
the coset $(\mcC_{i-1}\backslash\mcC_i)$. \textbf{Step 2.} For every label-0 edges, 
add a parallel label-1 edge. Then we get a trellis representing the cover set 
$\Delta(\mcC_{i-1}\backslash\mcC_i)$. But it is not a one-to-one trellis. 
\textbf{Step 3.} Let $T=(V,E,A)$ be the trellis we just constructed, use algorithm 1 
to convert it into a proper trellis $T^\ast=(V^\ast,E^\ast,A)$, where for 
$i=0,1,2,\cdots,\ell$,  vertices in $V^\ast_i$ are labeled uniquely by the subsets of 
$V_i$. $T^\ast$ will thus be a one-to-one trellis representing the same cover set 
$\Delta(\mcC_{i-1}\backslash\mcC_i)$.
 	
The proper trellis algorithm allows us to calculate the full polarization behavior of 
$K_{32}$, as shown in Figure~\ref{fig_B32}. Unfortunately, the computational 
complexity is still too high for $K_{64}$, in which we were able to explicitly 
calculate the erasure probability polynomials associated with the last and first 15 
rows in the kernel, as shown in Figure~\ref{fig:B64_15}.

\begin{figure}[t!]
 		\centering
		\includegraphics[height=7.5cm]{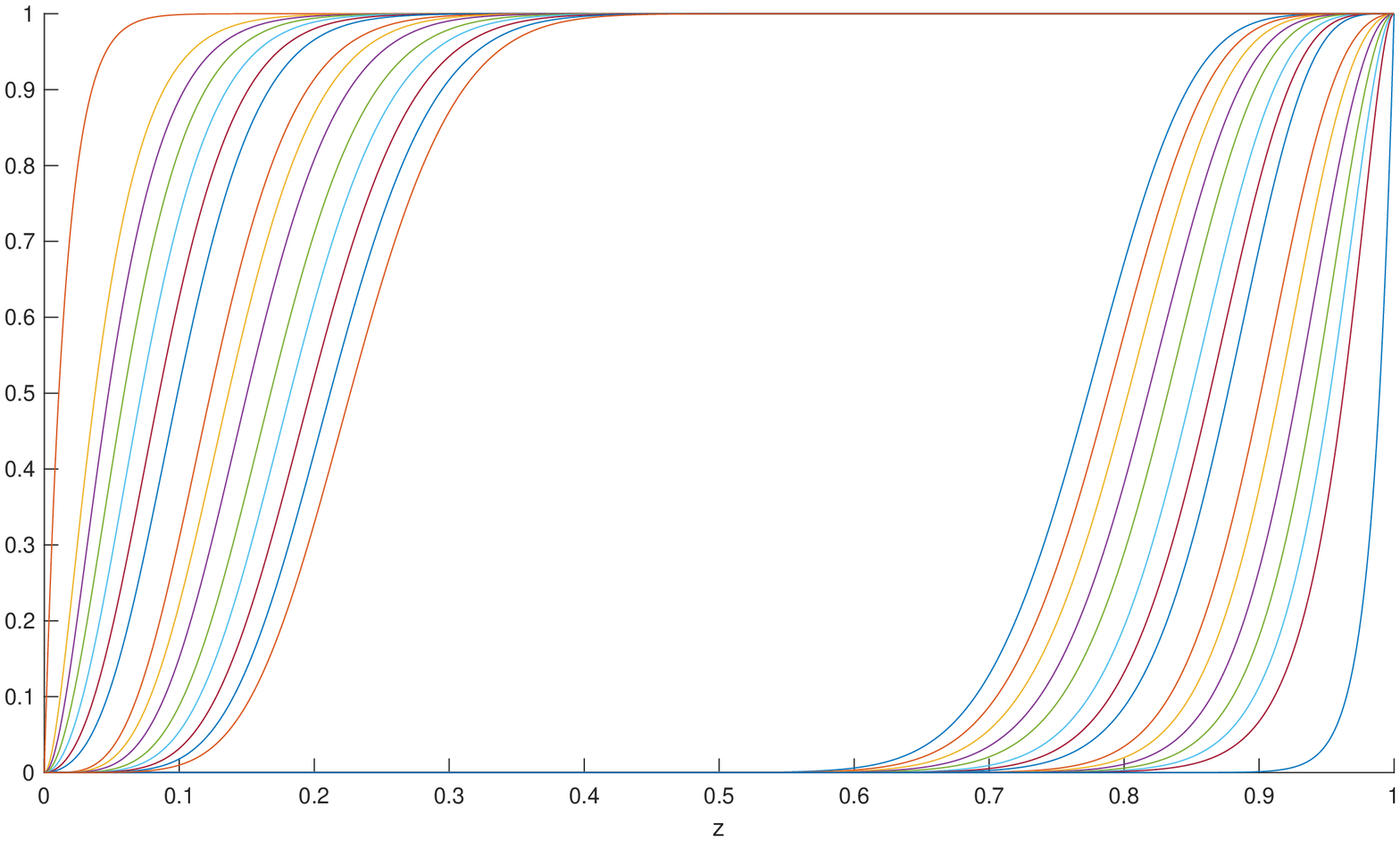}
 		\caption{Polarization behavior of the first and the last 15 rows of 
		kernel $K_{64}$}
 		\label{fig:B64_15}
\end{figure}

\begin{figure}[t!]
 		\centering
		\includegraphics[height=7.2cm]{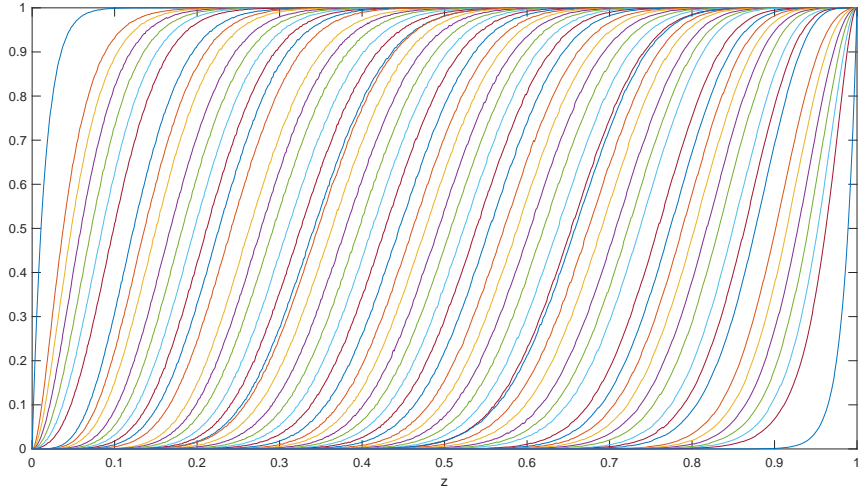}
 		\caption{Polarization behavior of kernel $K_{64}$ by Monte Carlo
		interpolation method, where $1000$ values of $z$ are evaluated
		uniformly among $[0,1]$ with $N=10^6$.}
 		\label{fig:B64_sim}
\end{figure}
 	
\vspace{3mm}
\noindent
\textbf{\emph{Stitching Trellis Algorithm}}
\vspace{3mm}
 	
\noindent 
The complexity of proper trellis algorithm depends on the number of vertices in the 
trellis. It's difficult to predict the number of vertices for general kernels, which 
could be significantly large. Hence, we also propose an alternative approach which 
also constructs a one-to-one trellis for $\Delta(\mcC_{i-1}\backslash\mcC_i)$, but 
has far less vertices. The stitching trellis algorithm differs from the proper trellis 
algorithm only by \textbf{Step 3}: Let $T=(V,E,A)$ be the trellis we just constructed, 
use algorithm 1 only for $i$ from 0 to $(\ell/2-1)$ to convert the first half of $T$ 
into a proper trellis $T_1$. Reverse algorithm 1 to convert the second half of $T$ 
into a \emph{co-proper} trellis $T_2$. Let $V_1,V_2$ be the vertex class of $T_1,T_2$ 
at time $\ell/2$. Connect $T_1$ and $T_2$ by adding an edge $(v_1,v_2)$ with label 0 
for every pair of vertices $v_1\in V_1,v_2\in V_2$ where $L(v_1)\cap L(v_2)\neq 0$. 
Then the combined trellis, called \emph{stitching trellis}, will be a one-to-one 
trellis representing the same cover set $\Delta(\mcC_{i-1}\backslash\mcC_i)$.
 	
The first and second half of the stitching trellis are proper and coproper 
respectively. Therefore, its number of vertices is bounded by $2^{\ell/2 +1}$, which 
is far less than a proper trellis. Unfortunately, the naive way of stitching the 
middle segment requires a large amount of computation. We are still searching for a 
method to reduce its complexity and we believe this can be of independent interest to 
other researchers as well. Assuming such an efficient stitching is in place, the 
stitching trellis will be much more efficient than the proper trellis, which can also 
be used in other applications.
 	
\begin{figure}[t]
 		\centering
 		
 		\makeatletter
 		\tikzset{nomorepostaction/.code={\let\tikz@postactions\pgfutil@empty}}
 		\makeatother
 		\begin{tikzpicture}[xscale=1.5,yscale = 1.5]
 		\begin{scope}[shift={(-4.5,0)},xscale=0.9,yscale=1,
 		ps/.style={
 			postaction={nomorepostaction,decorate,
 				decoration={markings,mark=at position 0.7 with {\arrow[#1]{stealth}}}
 			}
 		},
 		nd/.style={circle,fill=black}
 		]
 		\def\h{1}
 		\def\v{0.5}
 		\def\r{0.3}
 		
 		\draw [nd] (0,0) circle (\r ex);					
 		\draw [nd] (\h,\v) circle (\r ex);			
 		\draw [nd] (\h,-\v) circle (\r ex);			
 		
 		\draw [nd] (2*\h,\v) circle (\r ex);			
 		\draw [nd] (2*\h,-\v) circle (\r ex);		
 		\draw [nd] (2*\h,2*\v) circle (\r ex);		
 		\draw [nd] (2*\h,-2*\v) circle (\r ex);		
 		
 		\draw [nd] (3*\h,\v) circle (\r ex);		
 		\draw [nd] (3*\h,-\v) circle (\r ex);		
 		\draw [nd] (4*\h,0) circle (\r ex);			

 		\draw [ps] (0,0) -- node [above] {{\scriptsize 1}} (\h,\v);
 		\draw [ps] (0,0) -- node [below] {{\scriptsize 0}} (\h,-\v);
 		
 		\draw [ps] (\h,\v) -- node [below] {{\scriptsize 0}} (2*\h,\v);
 		\draw [ps] (\h,\v) -- node [above] {{\scriptsize 1}} (2*\h,2*\v);
 		\draw [ps] (\h,-\v) -- node [above] {{\scriptsize 1}} (2*\h,-\v);
 		\draw [ps] (\h,-\v) -- node [below] {{\scriptsize 0}} (2*\h,-2*\v);
 		
 		\draw [ps] (2*\h,\v) -- node [pos=0.3,above] {{\scriptsize 0}} (3*\h,-\v);
 		\draw [ps] (2*\h,-2*\v) -- node [below] {{\scriptsize 1}} (3*\h,-\v);
 		\draw [ps] (2*\h,-\v) -- node [pos=0.3,below] {{\scriptsize 1}} (3*\h,\v);
 		\draw [ps] (2*\h,2*\v) -- node [above] {{\scriptsize 0}} (3*\h,\v);
 		
 		\draw [ps] (3*\h,\v) -- node [above] {{\scriptsize 0}} (4*\h,0);
 		\draw [ps] (3*\h,-\v) -- node [below] {{\scriptsize 1}} (4*\h,0);
 		\end{scope}
 		\begin{scope}[xscale=0.9,yscale=1,
 		ps/.style={
 			postaction={nomorepostaction,decorate,
 				decoration={markings,mark=at position 0.7 with {\arrow[#1]{stealth}}}
 			}
 		},
 		nd/.style={circle,fill=black}
 		]
 		\def\h{1}
 		\def\v{0.5}
 		\def\r{0.3}
 		
 		\draw [nd] (0,0) circle (\r ex);					
 		\draw [nd] (\h,\v) circle (\r ex);			
 		\draw [nd] (\h,-\v) circle (\r ex);			
 		
 		\draw [nd] (2*\h,\v) circle (\r ex);			
 		\draw [nd] (2*\h,-\v) circle (\r ex);		
 		\draw [nd] (2*\h,2*\v) circle (\r ex);		
 		\draw [nd] (2*\h,-2*\v) circle (\r ex);		
 		
 		\draw [nd] (3*\h,\v) circle (\r ex);		
 		\draw [nd] (3*\h,-\v) circle (\r ex);		
 		\draw [nd] (4*\h,0) circle (\r ex);			
 		
 		\draw [ps] (0,0) -- node [above] {{\scriptsize 1}} (\h,\v);
 		\draw [ps] (0,0) -- node [below,pos=0.4,yshift=.3ex] {{\scriptsize 0}} (\h,-\v);
 		\draw [ps,color=red] (0,0) to [bend right=60,color=red] node [below] {{\scriptsize 1}} (\h,-\v);
 		\draw [ps] (\h,\v) -- node [below,yshift=.5ex] {{\scriptsize 0}} (2*\h,\v);
 		\draw [ps,color=red] (\h,\v)  to [bend right=60,color=red] node [below] {{\scriptsize 1}} (2*\h,\v);
 		
 		\draw [ps] (\h,\v) -- node [above] {{\scriptsize 1}} (2*\h,2*\v);
 		\draw [ps] (\h,-\v) -- node [above] {{\scriptsize 1}} (2*\h,-\v);
 		
 		\draw [ps] (\h,-\v) -- node [below,pos=0.4,yshift=.3ex] {{\scriptsize 0}} (2*\h,-2*\v);
 		\draw [ps,color=red] (\h,-\v)  to [bend right=60,color=red] node [below] {{\scriptsize 1}} (2*\h,-2*\v);
 		\draw [ps] (2*\h,\v) -- node [pos=0.2,below] {{\scriptsize 0}} (3*\h,-\v);
 		\draw [ps,color=red] (2*\h,\v) to [bend left] node [pos=0.2,above,color=red] {{\scriptsize 1}} (3*\h,-\v);
 		
 		\draw [ps] (2*\h,-2*\v) -- node [below] {{\scriptsize 1}} (3*\h,-\v);
 		\draw [ps] (2*\h,-\v) -- node [pos=0.3,below] {{\scriptsize 1}} (3*\h,\v);
 		
 		\draw [ps] (2*\h,2*\v) -- node [above,pos=0.6,yshift=-0.3ex] {{\scriptsize 0}} (3*\h,\v);
 		\draw [ps,color=red] (2*\h,2*\v) to [bend left=60,color=red] node [above] {{\scriptsize 1}} (3*\h,\v);
 		\draw [ps] (3*\h,\v) -- node [above,pos=0.6,yshift=-0.3ex] [above] {{\scriptsize 0}} (4*\h,0);
 		\draw [ps,color=red] (3*\h,\v) to [bend left=60,color=red] node [above] {{\scriptsize 1}} (4*\h,0);
 		\draw [ps] (3*\h,-\v) -- node [below] {{\scriptsize 1}} (4*\h,0);
 		\end{scope}
 		\begin{scope}[shift={(-4.5,-2.7)},xscale=0.9,yscale=1,
 		ps/.style={
 			postaction={nomorepostaction,decorate,
 				decoration={markings,mark=at position 0.7 with {\arrow[#1]{stealth}}}
 			}
 		},
 		nd/.style={circle,fill=black}
 		]
 		\def\h{1}
 		\def\v{0.5}
 		\def\r{0.3}
 		
 		\draw [nd] (0,0) circle (\r ex);					
 		\draw [nd] (\h,\v) circle (\r ex);			
 		\draw [nd] (\h,-\v) circle (\r ex);			
 		
 		\draw [nd] (2*\h,\v) circle (\r ex);			
 		\draw [nd] (2*\h,-\v) circle (\r ex);		
 		\draw [nd] (2*\h,2*\v) circle (\r ex);		
 		\draw [nd] (2*\h,-2*\v) circle (\r ex);		
 		
 		\draw [nd] (3*\h,\v) circle (\r ex);		
 		\draw [nd] (3*\h,-\v) circle (\r ex);		
 		\draw [nd] (4*\h,0) circle (\r ex);			

 		\draw [ps] (0,0) -- node [above] {{\scriptsize 0}} (\h,\v);
 		\draw [ps] (0,0) -- node [below] {{\scriptsize 1}} (\h,-\v);
 		
 		\draw [ps] (\h,\v) -- node [below] {{\scriptsize 1}} (2*\h,\v);
 		\draw [ps] (\h,\v) -- node [above] {{\scriptsize 0}} (2*\h,2*\v);
 		\draw [ps] (\h,-\v) -- node [above] {{\scriptsize 0}} (2*\h,-\v);
 		\draw [ps] (\h,-\v) -- node [below] {{\scriptsize 1}} (2*\h,-2*\v);
 		
 		\draw [ps] (2*\h,2*\v) -- node [above] {{\scriptsize 1}} (3*\h,\v);
 		\draw [ps] (2*\h,\v) -- node [pos=0.3,above] {{\scriptsize 1}} (3*\h,-\v);
 		\draw [ps] (2*\h,-\v) to [bend left=10] node [pos=0.1,above] {{\scriptsize 0}} (3*\h,\v);
 		\draw [ps] (2*\h,-\v) to [bend right=10] node [pos=0.3,below] {{\scriptsize 1}} (3*\h,\v);
 		\draw [ps] (2*\h,-2*\v) to [bend left=10] node [pos=0.2,above] {{\scriptsize 0}} (3*\h,-\v);
 		\draw [ps] (2*\h,-2*\v) to [bend right=10] node [below] {{\scriptsize 1}} (3*\h,-\v);
 		
 		\draw [ps] (3*\h,\v) -- node [above] {{\scriptsize 1}} (4*\h,0);
 		\draw [ps] (3*\h,-\v) to [bend left=10] node [pos=0.3,above] {{\scriptsize 0}} (4*\h,0);
 		\draw [ps] (3*\h,-\v) to [bend right=10] node [below] {{\scriptsize 1}} (4*\h,0);
 		\end{scope}
 		\begin{scope}[shift={(-0.35,-2.7)},xscale=0.9,yscale=1,
 		ps/.style={
 			postaction={nomorepostaction,decorate,
 				decoration={markings,mark=at position 0.7 with {\arrow[#1]{stealth}}}
 			}
 		},
 		nd/.style={circle,fill=black}
 		]
 		\def\h{1}
 		\def\v{0.5}
 		\def\r{0.3}
 		\def\dd{0.7}
 		
 		\draw [nd] (0,0) circle (\r ex);					
 		\draw [nd] (\h,\v) circle (\r ex);			
 		\draw [nd] (\h,-\v) circle (\r ex);			
 		
 		\draw [nd] (2*\h,\v) circle (\r ex);			
 		\draw [nd] (2*\h,-\v) circle (\r ex);		
 		\draw [nd] (2*\h,2*\v) circle (\r ex);		
 		\draw [nd] (2*\h,-2*\v) circle (\r ex);		
 		
 		\draw [nd] (2*\h+\dd,\v) circle (\r ex);			
 		\draw [nd] (2*\h+\dd,-\v) circle (\r ex);		
 		\draw [nd] (2*\h+\dd,2*\v) circle (\r ex);		
 		\draw [nd] (2*\h+\dd,-2*\v) circle (\r ex);		
 		
 		\draw [nd] (3*\h+\dd,\v) circle (\r ex);		
 		\draw [nd] (3*\h+\dd,-\v) circle (\r ex);		
 		\draw [nd] (4*\h+\dd,0) circle (\r ex);			

 		\draw [ps] (0,0) -- node [above] {{\scriptsize 0}} (\h,\v);
 		\draw [ps] (0,0) -- node [below] {{\scriptsize 1}} (\h,-\v);
 		
 		\draw [ps] (\h,\v) -- node [below] {{\scriptsize 0}} (2*\h,\v);
 		\draw [ps] (\h,\v) -- node [above] {{\scriptsize 1}} (2*\h,2*\v);
 		\draw [ps] (\h,-\v) -- node [above] {{\scriptsize 1}} (2*\h,-\v);
 		\draw [ps] (\h,-\v) -- node [below] {{\scriptsize 0}} (2*\h,-2*\v);
 		
 		\draw [ps] (2*\h,2*\v) -- (2*\h+\dd,2*\v);
 		\draw [ps] (2*\h,2*\v) -- (2*\h+\dd,-\v);
 		
 		\draw [ps] (2*\h,\v) -- (2*\h+\dd,-\v);
 		
 		\draw [ps] (2*\h,-\v) -- (2*\h+\dd,2*\v);
 		\draw [ps] (2*\h,-\v) -- (2*\h+\dd,\v);
 		\draw [ps] (2*\h,-\v) -- (2*\h+\dd,-\v);
 		\draw [ps] (2*\h,-\v) -- (2*\h+\dd,-2*\v);
 		
 		\draw [ps] (2*\h,-2*\v) -- (2*\h+\dd,-\v);
 		\draw [ps] (2*\h,-2*\v) -- (2*\h+\dd,-2*\v);
 		
 		\draw [ps] (2*\h+\dd,\v) -- node [below] {{\scriptsize 0}} (3*\h+\dd,\v);
 		\draw [ps] (2*\h+\dd,2*\v) -- node [above] {{\scriptsize 1}} (3*\h+\dd,\v);
 		\draw [ps] (2*\h+\dd,-\v) -- node [above] {{\scriptsize 1}} (3*\h+\dd,-\v);
 		\draw [ps] (2*\h+\dd,-2*\v) -- node [below] {{\scriptsize 0}} (3*\h+\dd,-\v);
 		
 		\draw [ps] (3*\h+\dd,\v) -- node [above] {{\scriptsize 0}} (4*\h+\dd,0);
 		\draw [ps] (3*\h+\dd,-\v) -- node [below] {{\scriptsize 1}} (4*\h+\dd,0);
 		
 		%
 		%
 		\end{scope}
 		\end{tikzpicture}
 		\caption{An example for the trellis algorithms. Top left: minimal trellis for 
		$\mcC_1\backslash\mcC_2$ in $K_4=K_2^{\otimes 2}$; top right: step 2 of 
		trellis algorithm; bottom left: step 3 of proper trellis algorithm; bottom 
		right: step 3 of the stitching algorithm.}
 		\label{fig_trellis}
\end{figure}
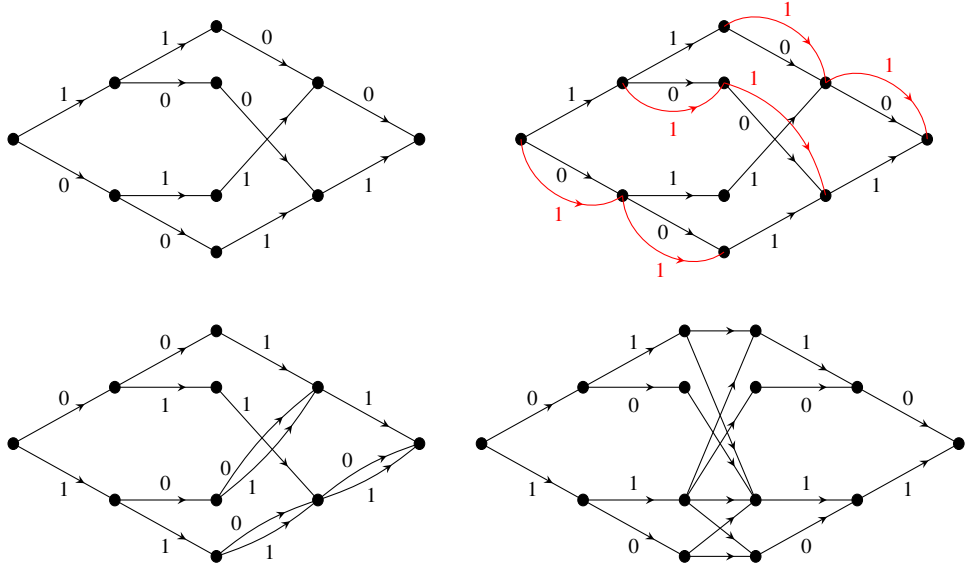
 	
\subsection{Monte Carlo Interpolation Method}
\noindent 
As discussed earlier, the complexity of the trellis-based algorithms grow too high for 
the intermediate bit-channels of $K_{64}$. We present a Monte Carlo algorithm to 
estimate the values of polynomials $f_i(z)$ for any given $z\in(0,1)$. We recall again 
that $f_i(z)$ denotes the erasure probability of the $i$-th bit-channel $W_i$ given 
that the communication is taking place over a BEC$(z)$. A naive yet explicit approach 
to formulate $f_i(z)$ is to cross check all $2^{\ell}$ erasure patterns to discover 
the exact ratio of which become uncorrectable from $W_i$'s point of view. Instead, we 
propose to take $N$ samples of such erasure patterns and estimate the ratio 
accordingly. We recall that the computational complexity of determining 
``correctability'' is no more than the complexity of a MAP decoder for the BEC, which 
is bounded by $O(\ell^\omega)$, where $\omega$ is the exponent of matrix 
multiplication. Therefore, the overall complexity of the proposed approximation 
method can be bounded by $O(N\ell^\omega)$. While this approach adds some uncertainty 
to our derivations, the numerical simulations suggest that  $\hat{f}_i(z)$'s for 
$\forall i$ become visibly smooth and stable at $N=10^6$, as shown in 
Figure~\ref{fig:B64_sim}. The estimated value of $\mu(K_{64}) \simeq 2.87$ is 
generated by invoking the recursive methods in~\cite{Hassani} initialized with 
$\hat{f}_i(z)$'s for $\forall i$.
 	
If the accurate values of $f_i(z)$ were known, one could use the bounding techniques 
in~\cite{Hassani} to show that

\begin{align}\label{eq:actual_bound_b0}
 	&\mu(K_{64}) \leq -\bigg({\log_{64}\bigg(\sup_{z\in(0,1)} \frac{\frac{1}{64}\sum_{i=1}^{64}g\big(f_i(z)\big)}{g(z)}\bigg)}\bigg)^{-1}
\end{align}

\noindent 
where $g(z)$ is a positive test function on $(0,1)$. However for kernel $K_{64}$, 
due to high computational complexity, 34 intermediate polarization bebavior polynomials 
are unknown. But we can still derive the strict upperbounds and lowerbounds for those
unknown $f_i(z)$s' to get the following theorem, with the proof in Appendix B.

\begin{theorem}
	\begin{align}
		\mu(K_{64})\le 2.9603
	\end{align}
\end{theorem}

 	
\section*{Acknowledgment}
\noindent 
We are grateful to Hamed Hassani and Peter Trifonov for very 
helpful discussions. We are also indebted to Peter Trifonov
for sharing the source code of his BDD program.

\begin{appendices}
\section{Kernels $K_{32}$ and $K_{64}$}
	\label{appendix_kernels}
 	 	\begin{figure}[h]
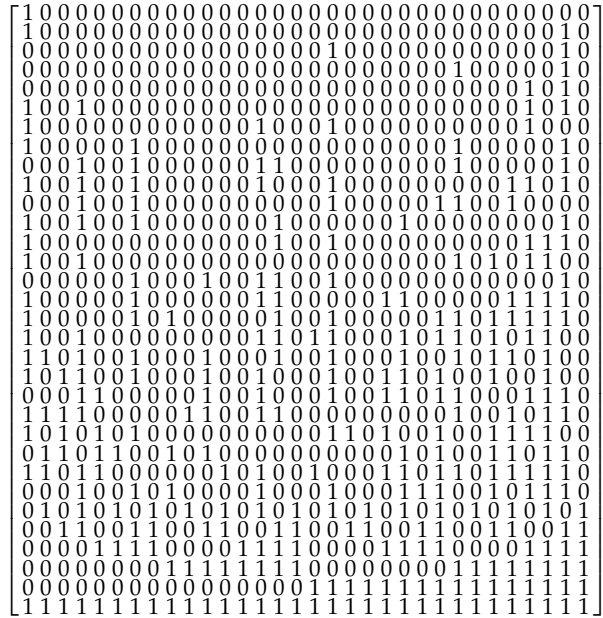

 	 		\centering
 	 		\scalebox{0.8}{
 	 			$\left[\arraycolsep=1.5pt\def\arraystretch{0.35}
 	 			\begin{array}{cccccccccccccccccccccccccccccccc} 
 	 			1& 0& 0& 0& 0& 0& 0& 0& 0& 0& 0& 0& 0& 0& 0& 0& 0& 0& 0& 0& 0& 0& 0& 0& 0& 0& 0& 0& 0& 0& 0& 0\\
 	 			1& 0& 0& 0& 0& 0& 0& 0& 0& 0& 0& 0& 0& 0& 0& 0& 0& 0& 0& 0& 0& 0& 0& 0& 0& 0& 0& 0& 0& 0& 1& 0\\ 
 	 			0& 0& 0& 0& 0& 0& 0& 0& 0& 0& 0& 0& 0& 0& 0& 0& 0& 1& 0& 0& 0& 0& 0& 0& 0& 0& 0& 0& 0& 0& 1& 0\\ 
 	 			0& 0& 0& 0& 0& 0& 0& 0& 0& 0& 0& 0& 0& 0& 0& 0& 0& 0& 0& 0& 0& 0& 0& 0& 1& 0& 0& 0& 0& 0& 1& 0\\
 	 			0& 0& 0& 0& 0& 0& 0& 0& 0& 0& 0& 0& 0& 0& 0& 0& 0& 0& 0& 0& 0& 0& 0& 0& 0& 0& 0& 0& 1& 0& 1& 0\\ 
 	 			1& 0& 0& 1& 0& 0& 0& 0& 0& 0& 0& 0& 0& 0& 0& 0& 0& 0& 0& 0& 0& 0& 0& 0& 0& 0& 0& 0& 1& 0& 1& 0\\ 
 	 			1& 0& 0& 0& 0& 0& 0& 0& 0& 0& 0& 0& 0& 1& 0& 0& 0& 1& 0& 0& 0& 0& 0& 0& 0& 0& 0& 0& 1& 0& 0& 0\\ 
 	 			1& 0& 0& 0& 0& 0& 1& 0& 0& 0& 0& 0& 0& 0& 0& 0& 0& 0& 0& 0& 0& 0& 0& 0& 1& 0& 0& 0& 0& 0& 1& 0\\ 
 	 			0& 0& 0& 1& 0& 0& 1& 0& 0& 0& 0& 0& 0& 1& 1& 0& 0& 0& 0& 0& 0& 0& 0& 0& 1& 0& 0& 0& 0& 0& 1& 0\\
 	 			1& 0& 0& 1& 0& 0& 1& 0& 0& 0& 0& 0& 0& 1& 0& 0& 0& 1& 0& 0& 0& 0& 0& 0& 0& 0& 0& 1& 1& 0& 1& 0\\
 	 			0& 0& 0& 1& 0& 0& 1& 0& 0& 0& 0& 0& 0& 0& 0& 0& 0& 1& 0& 0& 0& 0& 0& 1& 1& 0& 0& 1& 0& 0& 0& 0\\
 	 			1& 0& 0& 1& 0& 0& 1& 0& 0& 0& 0& 0& 0& 0& 1& 0& 0& 0& 0& 0& 0& 1& 0& 0& 0& 0& 0& 0& 0& 0& 1& 0\\
 	 			1& 0& 0& 0& 0& 0& 0& 0& 0& 0& 0& 0& 0& 0& 1& 0& 0& 1& 0& 0& 0& 0& 0& 0& 0& 0& 0& 0& 1& 1& 1& 0\\
 	 			1& 0& 0& 1& 0& 0& 0& 0& 0& 0& 0& 0& 0& 0& 0& 0& 0& 0& 0& 0& 0& 0& 0& 0& 1& 0& 1& 0& 1& 1& 0& 0\\
 	 			0& 0& 0& 0& 0& 0& 1& 0& 0& 0& 1& 0& 0& 1& 1& 0& 0& 1& 0& 0& 0& 0& 0& 0& 0& 0& 0& 0& 0& 0& 1& 0\\ 
 	 			1& 0& 0& 0& 0& 0& 1& 0& 0& 0& 0& 0& 0& 1& 1& 0& 0& 0& 0& 0& 1& 1& 0& 0& 0& 0& 0& 1& 1& 1& 1& 0\\ 
 	 			1& 0& 0& 0& 0& 0& 1& 0& 1& 0& 0& 0& 0& 0& 1& 0& 0& 1& 0& 0& 0& 0& 0& 1& 1& 0& 1& 1& 1& 1& 1& 0\\ 
 	 			1& 0& 0& 1& 0& 0& 0& 0& 0& 0& 0& 0& 0& 1& 1& 0& 1& 1& 0& 0& 0& 1& 0& 1& 1& 0& 1& 0& 1& 1& 0& 0\\ 
 	 			1& 1& 0& 1& 0& 0& 1& 0& 0& 0& 1& 0& 0& 0& 1& 0& 0& 1& 0& 0& 0& 1& 0& 0& 1& 0& 1& 1& 0& 1& 0& 0\\ 
 	 			1& 0& 1& 1& 0& 0& 1& 0& 0& 0& 1& 0& 0& 1& 0& 0& 0& 1& 0& 0& 1& 1& 0& 1& 0& 0& 1& 0& 0& 1& 0& 0\\
 	 			0& 0& 0& 1& 1& 0& 0& 0& 0& 0& 1& 0& 0& 1& 0& 0& 0& 1& 0& 0& 1& 1& 0& 1& 1& 0& 0& 0& 1& 1& 1& 0\\ 
 	 			1& 1& 1& 1& 0& 0& 0& 0& 0& 1& 1& 0& 0& 1& 1& 0& 0& 0& 0& 0& 0& 0& 0& 0& 1& 0& 0& 1& 0& 1& 1& 0\\ 
 	 			1& 0& 1& 0& 1& 0& 1& 0& 0& 0& 0& 0& 0& 0& 0& 0& 0& 1& 1& 0& 1& 0& 0& 1& 0& 0& 1& 1& 1& 1& 0& 0\\ 
 	 			0& 1& 1& 0& 1& 1& 0& 0& 1& 0& 1& 0& 0& 0& 0& 0& 0& 0& 0& 0& 0& 1& 0& 1& 0& 0& 1& 1& 0& 1& 1& 0\\
 	 			1& 1& 0& 1& 1& 0& 0& 0& 0& 0& 0& 1& 0& 1& 0& 0& 1& 0& 0& 0& 1& 1& 0& 1& 1& 0& 1& 1& 1& 1& 1& 0\\ 
 	 			0& 0& 0& 1& 0& 0& 1& 0& 1& 0& 0& 0& 0& 1& 0& 0& 0& 1& 0& 0& 0& 1& 1& 1& 0& 0& 1& 0& 1& 1& 1& 0\\ 
 	 			0& 1& 0& 1& 0& 1& 0& 1& 0& 1& 0& 1& 0& 1& 0& 1& 0& 1& 0& 1& 0& 1& 0& 1& 0& 1& 0& 1& 0& 1& 0& 1\\ 
 	 			0& 0& 1& 1& 0& 0& 1& 1& 0& 0& 1& 1& 0& 0& 1& 1& 0& 0& 1& 1& 0& 0& 1& 1& 0& 0& 1& 1& 0& 0& 1& 1\\
 	 			0& 0& 0& 0& 1& 1& 1& 1& 0& 0& 0& 0& 1& 1& 1& 1& 0& 0& 0& 0& 1& 1& 1& 1& 0& 0& 0& 0& 1& 1& 1& 1\\ 
 	 			0& 0& 0& 0& 0& 0& 0& 0& 1& 1& 1& 1& 1& 1& 1& 1& 0& 0& 0& 0& 0& 0& 0& 0& 1& 1& 1& 1& 1& 1& 1& 1\\  
 	 			0& 0& 0& 0& 0& 0& 0& 0& 0& 0& 0& 0& 0& 0& 0& 0& 1& 1& 1& 1& 1& 1& 1& 1& 1& 1& 1& 1& 1& 1& 1& 1\\ 
 	 			1& 1& 1& 1& 1& 1& 1& 1& 1& 1& 1& 1& 1& 1& 1& 1& 1& 1& 1& 1& 1& 1& 1& 1& 1& 1& 1& 1& 1& 1& 1& 1\\ 
 	 			\end{array}\right]$
 	 		}
 	 		\caption{Kernel $K_{32}$}
			\label{fig:K32}
 	 	\end{figure}
\begin{figure}[p]
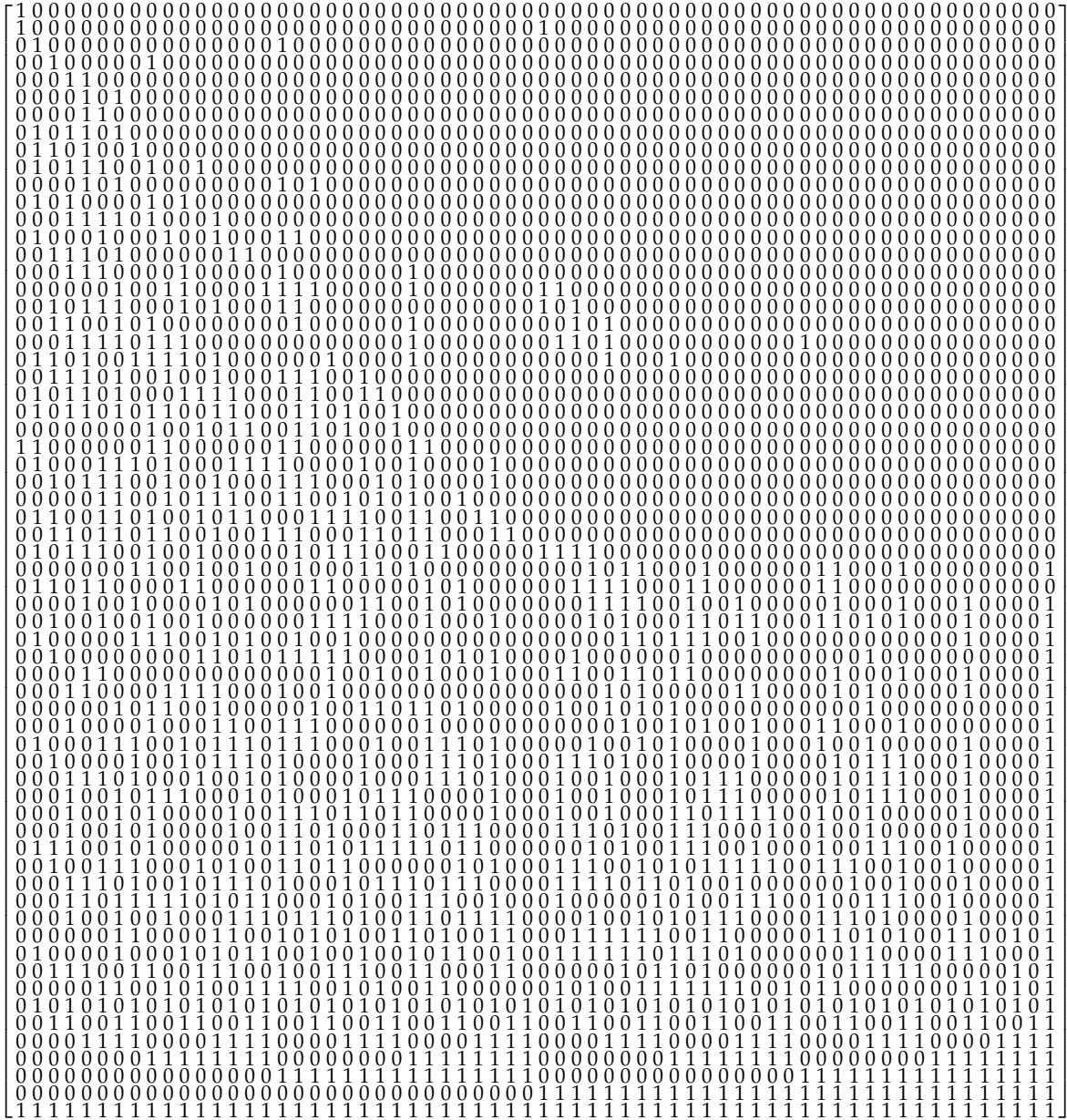

	\centering
	\scalebox{0.8}{
		$\left[\arraycolsep=1.5pt\def\arraystretch{0.35}
		\begin{array}{cccccccccccccccccccccccccccccccccccccccccccccccccccccccccccccccc}
		1&0&0&0&0&0&0&0&0&0&0&0&0&0&0&0&0&0&0&0&0&0&0&0&0&0&0&0&0&0&0&0&0&0&0&0&0&0&0&0&0&0&0&0&0&0&0&0&0&0&0&0&0&0&0&0&0&0&0&0&0&0&0&0\\
		1&0&0&0&0&0&0&0&0&0&0&0&0&0&0&0&0&0&0&0&0&0&0&0&0&0&0&0&0&0&0&0&1&0&0&0&0&0&0&0&0&0&0&0&0&0&0&0&0&0&0&0&0&0&0&0&0&0&0&0&0&0&0&0\\
		0&1&0&0&0&0&0&0&0&0&0&0&0&0&0&0&1&0&0&0&0&0&0&0&0&0&0&0&0&0&0&0&0&0&0&0&0&0&0&0&0&0&0&0&0&0&0&0&0&0&0&0&0&0&0&0&0&0&0&0&0&0&0&0\\
		0&0&1&0&0&0&0&0&1&0&0&0&0&0&0&0&0&0&0&0&0&0&0&0&0&0&0&0&0&0&0&0&0&0&0&0&0&0&0&0&0&0&0&0&0&0&0&0&0&0&0&0&0&0&0&0&0&0&0&0&0&0&0&0\\
		0&0&0&1&1&0&0&0&0&0&0&0&0&0&0&0&0&0&0&0&0&0&0&0&0&0&0&0&0&0&0&0&0&0&0&0&0&0&0&0&0&0&0&0&0&0&0&0&0&0&0&0&0&0&0&0&0&0&0&0&0&0&0&0\\
		0&0&0&0&1&0&1&0&0&0&0&0&0&0&0&0&0&0&0&0&0&0&0&0&0&0&0&0&0&0&0&0&0&0&0&0&0&0&0&0&0&0&0&0&0&0&0&0&0&0&0&0&0&0&0&0&0&0&0&0&0&0&0&0\\
		0&0&0&0&1&1&0&0&0&0&0&0&0&0&0&0&0&0&0&0&0&0&0&0&0&0&0&0&0&0&0&0&0&0&0&0&0&0&0&0&0&0&0&0&0&0&0&0&0&0&0&0&0&0&0&0&0&0&0&0&0&0&0&0\\
		0&1&0&1&1&0&1&0&0&0&0&0&0&0&0&0&0&0&0&0&0&0&0&0&0&0&0&0&0&0&0&0&0&0&0&0&0&0&0&0&0&0&0&0&0&0&0&0&0&0&0&0&0&0&0&0&0&0&0&0&0&0&0&0\\
		0&1&1&0&1&0&0&1&0&0&0&0&0&0&0&0&0&0&0&0&0&0&0&0&0&0&0&0&0&0&0&0&0&0&0&0&0&0&0&0&0&0&0&0&0&0&0&0&0&0&0&0&0&0&0&0&0&0&0&0&0&0&0&0\\
		0&1&0&1&1&1&0&0&1&0&0&1&0&0&0&0&0&0&0&0&0&0&0&0&0&0&0&0&0&0&0&0&0&0&0&0&0&0&0&0&0&0&0&0&0&0&0&0&0&0&0&0&0&0&0&0&0&0&0&0&0&0&0&0\\
		0&0&0&0&1&0&1&0&0&0&0&0&0&0&0&0&1&0&1&0&0&0&0&0&0&0&0&0&0&0&0&0&0&0&0&0&0&0&0&0&0&0&0&0&0&0&0&0&0&0&0&0&0&0&0&0&0&0&0&0&0&0&0&0\\
		0&1&0&1&0&0&0&0&1&0&1&0&0&0&0&0&0&0&0&0&0&0&0&0&0&0&0&0&0&0&0&0&0&0&0&0&0&0&0&0&0&0&0&0&0&0&0&0&0&0&0&0&0&0&0&0&0&0&0&0&0&0&0&0\\
		0&0&0&1&1&1&1&0&1&0&0&0&1&0&0&0&0&0&0&0&0&0&0&0&0&0&0&0&0&0&0&0&0&0&0&0&0&0&0&0&0&0&0&0&0&0&0&0&0&0&0&0&0&0&0&0&0&0&0&0&0&0&0&0\\
		0&1&0&0&0&1&0&0&0&1&0&0&1&0&0&0&1&1&0&0&0&0&0&0&0&0&0&0&0&0&0&0&0&0&0&0&0&0&0&0&0&0&0&0&0&0&0&0&0&0&0&0&0&0&0&0&0&0&0&0&0&0&0&0\\
		0&0&1&1&1&0&1&0&0&0&0&0&0&1&1&0&0&0&0&0&0&0&0&0&0&0&0&0&0&0&0&0&0&0&0&0&0&0&0&0&0&0&0&0&0&0&0&0&0&0&0&0&0&0&0&0&0&0&0&0&0&0&0&0\\
		0&0&0&1&1&1&0&0&0&0&1&0&0&0&0&0&1&0&0&0&0&0&0&0&1&0&0&0&0&0&0&0&0&0&0&0&0&0&0&0&0&0&0&0&0&0&0&0&0&0&0&0&0&0&0&0&0&0&0&0&0&0&0&0\\
		0&0&0&0&0&0&1&0&0&1&1&0&0&0&0&1&1&1&1&0&0&0&0&0&1&0&0&0&0&0&0&0&1&1&0&0&0&0&0&0&0&0&0&0&0&0&0&0&0&0&0&0&0&0&0&0&0&0&0&0&0&0&0&0\\
		0&0&1&0&1&1&1&0&0&0&1&0&1&0&0&0&1&1&0&0&0&0&0&0&0&0&0&0&0&0&0&0&1&0&1&0&0&0&0&0&0&0&0&0&0&0&0&0&0&0&0&0&0&0&0&0&0&0&0&0&0&0&0&0\\
		0&0&1&1&0&0&1&0&1&0&0&0&0&0&0&0&0&1&0&0&0&0&0&0&1&0&0&0&0&0&0&0&0&0&1&0&1&0&0&0&0&0&0&0&0&0&0&0&0&0&0&0&0&0&0&0&0&0&0&0&0&0&0&0\\
		0&0&0&1&1&1&1&0&1&1&1&0&0&0&0&0&0&0&0&0&0&0&0&0&1&0&0&0&0&0&0&0&0&1&1&0&1&0&0&0&0&0&0&0&0&0&0&0&1&0&0&0&0&0&0&0&0&0&0&0&0&0&0&0\\
		0&1&1&0&1&0&0&1&1&1&1&0&1&0&0&0&0&0&0&1&0&0&0&0&1&0&0&0&0&0&0&0&0&0&0&0&1&0&0&0&1&0&0&0&0&0&0&0&0&0&0&0&0&0&0&0&0&0&0&0&0&0&0&0\\
		0&0&1&1&1&0&1&0&0&1&0&0&1&0&0&0&1&1&1&0&0&1&0&0&0&0&0&0&0&0&0&0&0&0&0&0&0&0&0&0&0&0&0&0&0&0&0&0&0&0&0&0&0&0&0&0&0&0&0&0&0&0&0&0\\
		0&1&0&1&1&0&1&0&0&0&1&1&1&1&0&0&0&1&1&0&0&1&1&0&0&0&0&0&0&0&0&0&0&0&0&0&0&0&0&0&0&0&0&0&0&0&0&0&0&0&0&0&0&0&0&0&0&0&0&0&0&0&0&0\\
		0&1&0&1&1&0&1&0&1&1&0&0&1&1&0&0&0&1&1&0&1&0&0&1&0&0&0&0&0&0&0&0&0&0&0&0&0&0&0&0&0&0&0&0&0&0&0&0&0&0&0&0&0&0&0&0&0&0&0&0&0&0&0&0\\
		0&0&0&0&0&0&0&0&1&0&0&1&0&1&1&0&0&1&1&0&1&0&0&1&0&0&0&0&0&0&0&0&0&0&0&0&0&0&0&0&0&0&0&0&0&0&0&0&0&0&0&0&0&0&0&0&0&0&0&0&0&0&0&0\\
		1&1&0&0&0&0&0&0&1&1&0&0&0&0&0&0&1&1&0&0&0&0&0&0&1&1&0&0&0&0&0&0&0&0&0&0&0&0&0&0&0&0&0&0&0&0&0&0&0&0&0&0&0&0&0&0&0&0&0&0&0&0&0&0\\
		0&1&0&0&0&1&1&1&0&1&0&0&0&1&1&1&1&0&0&0&0&1&0&0&1&0&0&0&0&1&0&0&0&0&0&0&0&0&0&0&0&0&0&0&0&0&0&0&0&0&0&0&0&0&0&0&0&0&0&0&0&0&0&0\\
		0&0&1&0&1&1&1&0&0&1&0&0&1&0&0&0&1&1&1&0&0&0&1&0&1&0&0&0&0&1&0&0&0&0&0&0&0&0&0&0&0&0&0&0&0&0&0&0&0&0&0&0&0&0&0&0&0&0&0&0&0&0&0&0\\
		0&0&0&0&0&1&1&0&0&1&0&1&1&1&0&0&1&1&0&0&1&0&1&0&1&0&0&1&0&0&0&0&0&0&0&0&0&0&0&0&0&0&0&0&0&0&0&0&0&0&0&0&0&0&0&0&0&0&0&0&0&0&0&0\\
		0&1&1&0&0&1&1&0&1&0&0&1&0&1&1&0&0&0&1&1&1&1&0&0&1&1&0&0&1&1&0&0&0&0&0&0&0&0&0&0&0&0&0&0&0&0&0&0&0&0&0&0&0&0&0&0&0&0&0&0&0&0&0&0\\
		0&0&1&1&0&1&1&0&1&1&0&0&1&0&0&1&1&1&0&0&0&1&1&0&1&1&0&0&0&1&1&0&0&0&0&0&0&0&0&0&0&0&0&0&0&0&0&0&0&0&0&0&0&0&0&0&0&0&0&0&0&0&0&0\\
		0&1&0&1&1&1&0&0&1&0&0&1&0&0&0&0&0&1&0&1&1&1&0&0&0&1&1&0&0&0&0&0&1&1&1&1&0&0&0&0&0&0&0&0&0&0&0&0&0&0&0&0&0&0&0&0&0&0&0&0&0&0&0&0\\
		0&0&0&0&0&0&0&1&1&0&0&1&0&0&1&0&0&1&0&0&0&1&1&0&1&0&0&0&0&0&0&0&0&0&0&1&0&1&1&0&0&0&1&0&0&0&0&0&0&1&1&0&0&0&1&0&0&0&0&0&0&0&0&1\\
		0&1&1&0&1&1&0&0&0&0&1&1&0&0&0&0&0&0&1&1&0&0&0&0&0&1&0&1&0&0&0&0&0&0&1&1&1&1&0&0&0&1&1&0&0&0&0&0&0&1&1&0&0&0&0&0&0&0&0&0&0&0&0&0\\
		0&0&0&0&1&0&0&1&0&0&0&0&1&0&1&0&0&0&0&0&0&1&1&0&0&1&0&1&0&0&0&0&0&0&0&1&1&1&1&0&0&1&0&0&1&0&0&0&0&0&1&0&0&0&1&0&0&0&1&0&0&0&0&1\\
		0&0&1&0&0&1&0&0&1&0&0&1&0&0&0&0&0&0&1&1&1&1&0&0&0&1&0&0&0&1&0&0&0&0&0&1&0&1&0&0&0&1&1&0&1&1&0&0&0&1&1&0&1&0&1&0&0&0&1&0&0&0&0&1\\
		0&1&0&0&0&0&0&1&1&1&0&0&1&0&1&0&0&1&0&0&1&0&0&0&0&0&0&0&0&0&0&0&0&0&0&0&0&1&1&0&1&1&1&0&0&1&0&0&0&0&0&0&0&0&0&0&0&0&1&0&0&0&0&1\\
		0&0&1&0&0&0&0&0&0&0&0&1&1&0&1&0&1&1&1&1&1&0&0&0&0&1&0&1&0&1&0&0&0&0&1&0&0&0&0&0&0&1&0&0&0&0&0&0&0&0&0&0&1&0&0&0&0&0&0&0&0&0&0&1\\
		0&0&0&0&1&1&0&0&0&0&0&0&0&0&0&0&0&0&0&1&0&0&1&0&0&1&0&0&0&1&0&0&0&1&1&0&0&1&1&0&1&1&0&0&0&0&0&0&0&0&1&0&0&0&1&0&0&0&1&0&0&0&0&1\\
		0&0&0&1&1&0&0&0&0&1&1&1&1&0&0&0&1&0&0&1&0&0&0&0&0&0&0&0&0&0&0&0&0&0&0&0&1&0&1&0&0&0&0&0&1&1&0&0&0&0&1&0&1&0&0&0&0&0&1&0&0&0&0&1\\
		0&0&0&0&0&0&1&0&1&1&0&0&1&0&0&0&0&0&1&0&0&1&1&0&1&1&0&1&0&0&0&0&0&1&0&0&1&0&1&0&1&0&0&0&0&0&0&0&0&0&0&0&1&0&0&0&0&0&0&0&0&0&0&1\\
		0&0&0&1&0&0&0&0&1&0&0&0&1&1&0&0&1&1&1&0&0&0&0&0&0&1&0&0&0&0&0&0&0&0&0&0&0&1&0&0&1&0&1&0&0&1&0&0&0&1&1&0&0&0&1&0&0&0&0&0&0&0&0&1\\
		0&1&0&0&0&1&1&1&0&0&1&0&1&1&1&0&1&1&1&0&0&0&1&0&0&1&1&1&0&1&0&0&0&0&0&1&0&0&1&0&1&0&0&0&0&1&0&0&0&1&0&0&1&0&0&0&0&0&1&0&0&0&0&1\\
		0&0&1&0&0&0&0&1&0&0&1&0&1&1&1&0&1&0&0&0&0&1&0&0&0&1&1&1&0&1&0&0&0&1&1&1&0&1&0&0&1&0&0&0&0&1&0&0&0&0&1&0&1&1&1&0&0&0&1&0&0&0&0&1\\
		0&0&0&1&1&1&0&1&0&0&0&1&0&0&1&0&1&0&0&0&0&1&0&0&0&1&1&1&0&1&0&0&0&1&0&0&1&0&0&0&1&0&1&1&1&0&0&0&0&0&1&0&1&1&1&0&0&0&1&0&0&0&0&1\\
		0&0&0&1&0&0&1&0&1&1&1&0&0&0&1&0&1&0&0&0&1&0&1&1&1&0&0&0&0&1&0&0&0&1&0&0&1&0&0&0&1&0&1&1&1&0&0&0&0&0&1&0&1&1&1&0&0&0&1&0&0&0&0&1\\
		0&0&0&1&0&0&1&0&1&0&0&0&0&1&0&0&1&1&1&0&1&1&0&1&1&0&0&0&0&1&0&0&0&1&0&0&1&0&0&0&1&1&0&1&1&1&1&0&0&1&0&0&1&0&0&0&0&0&1&0&0&0&0&1\\
		0&0&0&1&0&0&1&0&1&0&0&0&0&1&0&0&1&1&0&1&0&0&0&1&1&0&1&1&1&0&0&0&0&1&1&1&0&1&0&0&1&1&1&0&0&0&1&0&0&1&0&0&1&0&0&0&0&0&1&0&0&0&0&1\\
		0&1&1&1&0&0&1&0&1&0&0&0&0&0&1&0&1&1&0&1&0&1&1&1&1&1&0&1&1&0&0&0&0&0&0&1&0&1&0&0&1&1&1&0&0&1&0&0&0&1&0&0&1&1&1&0&0&1&0&0&0&0&0&1\\
		0&0&1&0&0&1&1&1&0&0&0&1&0&1&0&0&1&1&0&1&1&0&0&0&0&0&0&1&0&1&0&0&0&1&1&1&0&0&1&0&1&0&1&1&1&1&1&0&0&1&1&1&0&0&1&0&0&1&0&0&0&0&0&1\\
		0&0&0&1&1&1&0&1&0&0&1&0&1&1&1&0&1&0&0&0&1&0&1&1&1&0&1&1&1&0&0&0&0&1&1&1&1&0&1&1&0&1&0&0&1&0&0&0&0&0&0&1&0&0&1&0&0&0&1&0&0&0&0&1\\
		0&0&0&1&1&0&1&1&1&1&1&0&1&0&1&1&0&0&0&1&0&1&0&0&1&1&1&0&0&1&0&0&0&1&0&0&0&0&0&1&0&1&0&0&1&1&1&0&0&1&0&0&1&1&1&0&0&1&0&0&0&0&0&1\\
		0&0&0&1&0&0&1&0&0&1&0&0&0&1&1&1&0&1&1&1&0&1&0&0&1&1&0&1&1&1&1&0&0&0&0&1&0&0&1&0&1&0&1&1&1&0&0&0&0&1&1&1&0&1&0&0&0&0&1&0&0&0&0&1\\
		0&0&0&0&0&0&1&1&0&0&0&0&1&1&0&0&1&0&1&0&1&0&0&1&1&0&1&0&0&1&1&0&0&0&1&1&1&1&1&1&0&0&1&1&0&0&0&0&0&1&1&0&1&0&1&0&0&1&1&0&0&1&0&1\\
		0&1&0&0&0&0&1&0&0&0&1&0&1&0&1&1&0&0&1&0&0&1&0&0&1&0&1&1&0&0&1&0&0&1&1&1&1&1&1&0&1&1&1&0&1&0&0&0&0&0&0&1&1&0&0&0&0&1&1&1&0&0&0&1\\
		0&0&1&1&1&0&0&1&1&0&0&1&1&1&0&0&1&0&0&1&1&1&0&0&1&1&0&0&0&1&1&0&0&0&0&0&0&1&0&1&1&0&1&0&0&0&0&0&0&1&0&1&1&1&1&1&0&0&0&0&0&1&0&1\\
		0&0&0&0&0&1&1&0&0&1&0&1&0&0&1&1&1&1&0&0&1&0&1&0&0&1&1&0&0&0&0&0&0&1&0&1&0&0&1&1&1&1&1&1&1&0&0&1&0&1&1&0&0&0&0&0&0&0&1&1&0&1&0&1\\
		0&1&0&1&0&1&0&1&0&1&0&1&0&1&0&1&0&1&0&1&0&1&0&1&0&1&0&1&0&1&0&1&0&1&0&1&0&1&0&1&0&1&0&1&0&1&0&1&0&1&0&1&0&1&0&1&0&1&0&1&0&1&0&1\\
		0&0&1&1&0&0&1&1&0&0&1&1&0&0&1&1&0&0&1&1&0&0&1&1&0&0&1&1&0&0&1&1&0&0&1&1&0&0&1&1&0&0&1&1&0&0&1&1&0&0&1&1&0&0&1&1&0&0&1&1&0&0&1&1\\
		0&0&0&0&1&1&1&1&0&0&0&0&1&1&1&1&0&0&0&0&1&1&1&1&0&0&0&0&1&1&1&1&0&0&0&0&1&1&1&1&0&0&0&0&1&1&1&1&0&0&0&0&1&1&1&1&0&0&0&0&1&1&1&1\\
		0&0&0&0&0&0&0&0&1&1&1&1&1&1&1&1&0&0&0&0&0&0&0&0&1&1&1&1&1&1&1&1&0&0&0&0&0&0&0&0&1&1&1&1&1&1&1&1&0&0&0&0&0&0&0&0&1&1&1&1&1&1&1&1\\
		0&0&0&0&0&0&0&0&0&0&0&0&0&0&0&0&1&1&1&1&1&1&1&1&1&1&1&1&1&1&1&1&0&0&0&0&0&0&0&0&0&0&0&0&0&0&0&0&1&1&1&1&1&1&1&1&1&1&1&1&1&1&1&1\\
		0&0&0&0&0&0&0&0&0&0&0&0&0&0&0&0&0&0&0&0&0&0&0&0&0&0&0&0&0&0&0&0&1&1&1&1&1&1&1&1&1&1&1&1&1&1&1&1&1&1&1&1&1&1&1&1&1&1&1&1&1&1&1&1\\
		1&1&1&1&1&1&1&1&1&1&1&1&1&1&1&1&1&1&1&1&1&1&1&1&1&1&1&1&1&1&1&1&1&1&1&1&1&1&1&1&1&1&1&1&1&1&1&1&1&1&1&1&1&1&1&1&1&1&1&1&1&1&1&1
		\end{array}\right]$
	}
	\caption{Kernel $K_{64}$}
	\label{fig:K64}
\end{figure}

\section{Proof for Theorem 3}
Given an $\ell\times\ell$ kernel with polarization behavior $\{f_1,\cdots,f_\ell\}$,
for a fixed $z\in[0,1]$, we can define the process 
\begin{equation}
	Z_0 = z,\qquad Z_{n+1}=
	\begin{cases}
		f_1(Z_n) & \text{w.p. }1/\ell\\
		f_2(Z_n) & \text{w.p. }1/\ell\\
		\vdots & \vdots \\
		f_1(Z_n) & \text{w.p. }1/\ell\\
	\end{cases}
\end{equation}
First lets recall the scaling assumption
\begin{assumption}
	There exists $\mu\in(0,\infty)$ such that, for any $z,a,b\in(0,1)$ such that $a<b$,
	The limit $\lim_{n\rightarrow\infty}\ell^{\frac{n}{\mu}}\text{Pr}(Z_n\in[a,b])$ 
	exists in $(0,\infty)$.
\end{assumption}
\noindent For a generic test function $g:[0,1]\rightarrow[0,1]$, define the 
sequence of functions $\{g_n\}_{n\in\mathbb{N}}$ as $g_n:[0,1]\rightarrow[0,1]$ that
\begin{equation}
	g_n(z)=\mathbb{E}[g(Z_n)\mid Z_0=z]
\end{equation}
Then this sequence of functions satisfies the recursive relation
\begin{equation}
	g_0(z)=g(z),\qquad g_{n+1}(z)=\frac{1}{\ell}\sum_{i=1}^\ell g_n(f_i(z))
\end{equation}
Our approach of bounding $\mu(K_{64})$ has the following steps: (1) Find a suitable
test function $g(z)$; (2) provide an upperbound on the polarizing speed of the sequence
$\{g_n(z)\}_{n\in\mathbb{N}}$ and (3) turn this upperbound into bound for $\mu(K_{64})$.
We here define the sequence $\{b_n\}_{n\in\mathbb{N}}$ to measure the polarizing speed
of $\{g_n(z)\}_{n\in\mathbb{N}}$.
\begin{definition}
\begin{equation*}
	b_n(z)=\frac{g_{n+1}(z)}{g_n(z)},\qquad b_n=\sup_{z\in(0,1)}b_n(z)
\end{equation*}
\end{definition}
\noindent We can prove that $\{b_n\}_{n\in\mathbb{N}}$ is a decreasing sequence.
\begin{lemma}
	$\{b_n\}_{n\in\mathbb{N}}$ is a decreasing sequence.
\end{lemma}
\begin{proof}
	Since for any fixed $z$
	\begin{align*}
		g_{n+1}(z)&=\frac{1}{\ell}\sum_{i=1}^\ell g_n(f_i(z))\\
		&\le b_{n-1}\left(\frac{1}{\ell}\sum_{i=1}^\ell g_{n-1}(f_i(z))\right)\\
		&= b_{n-1}\cdot g_n(z)
	\end{align*}
	We have $b_n(z)=\frac{g_{n+1}(z)}{g_n(z)}\le b_{n-1}$ for any $z$. Therefore
	$b_n\le b_{n-1}$ and $\{b_n\}_{n\in\mathbb{N}}$ is a decreasing sequence. 
\end{proof}
\noindent From the above lemma we have: 
\begin{align*}
	g_n(z)&\le b_{n-1}g_{n-1}(z)\\
	&\le b_{0}g_{n-1}(z)\\
	&\le \cdots\\
	&\le b^n_{0}g(z)
\end{align*}
Next we use $b_0$ to give an upperbound for the scaling exponent of the kernel.
\begin{lemma}
	For $a,b\in(0,1)$ and $n\in\mathbb{N}$ we have:
	\begin{equation*}
		\frac{1}{n}\log_\ell\text{Pr}(Z_n\in[a,b])\le \log_\ell b_0+
		O\left(\frac{1}{n}\right)
	\end{equation*}
\end{lemma}
\begin{proof}
	By Markov inequality
	\begin{align*}
		\text{Pr}(Z_n\in[a,b])&\le \text{Pr}(g(Z_n)\ge \min_{z\in[a,b]}g(z))\\
		&\le \frac{\mathbb{E}[g(Z_n)]}{\min_{z\in[a,b]}g(z)}
	\end{align*}
	So                                                                                      
	\begin{align*}                                                                          
	\frac{1}{n}\log_\ell\text{Pr}(Z_n\in[a,b])&\le \frac{1}{n}\log_\ell\frac{(b_0)^{n}g(z)}{\min_{z\in[a,b]}g(z)}\\
	&\le \log_\ell b_0+\frac{1}{n}\left(\log_\ell\frac{g(z)}{\min_{z\in[a,b]}g(z)}\right)
	\end{align*}                                                                            
\end{proof}
\noindent Since by scaling assumption                                                             
\begin{align*}                                                                          
-\frac{1}{\mu}=\lim_{n\rightarrow\infty} \frac{1}{n}\log_\ell(\text{Pr}(Z_n\in[a,b]))   
\end{align*}                                                                            
We have $\mu\le -\frac{1}{\log_\ell b_0}$. We next pick the appropriate test function
$g$ and use $b_0$ to obtain a valid upperbound for $\mu(K_{64})$. First we explain in
detail how we construct this test function $g(z)$.

By the trellis algorithm, we get the explicit polynomials $f_{50},\cdots,f_{64}$.
By the duality theorem, we also get to know the explicit formulas for $f_1,\cdots,f_{15}$. And
there are 34 polarization behavior polynomials left unknown. But for those unknown
coefficients, we can calculate their upperbound as follows:
\begin{lemma}
	Let $A_0,A_1,\cdots,A_\ell$ be the weight enumerators for the coset 
	$(C_{i-1}\backslash C_i)$, then
	\begin{equation*}
		E_{i,w}\leq \min\left(\sum_{j=1}^i\binom{\ell-i}{j-i}A_i,\binom{\ell}{i}\right)
	\end{equation*}
\end{lemma}
\begin{proof}
	By theorem \ref{delta} any erasure pattern is uncorrectable iff it covers a codeword
	in $(C_{i-1}\backslash C_i)$. For each codeword in $(C_{i-1}\backslash C_i)$ of 
	weight $j$, there are $\binom{n-j}{w-j}$ erasure patterns with weight $w$ that
	covers it. So $E_{i,w}\le \sum_{j=1}^i\binom{\ell-i}{j-i}A_i$. On the other hand,
	$E_{i,w}$ is at most $\binom{\ell}{i}$.
\end{proof}
\begin{figure}[t]
	\begin{center}
		\includegraphics[width=8cm]{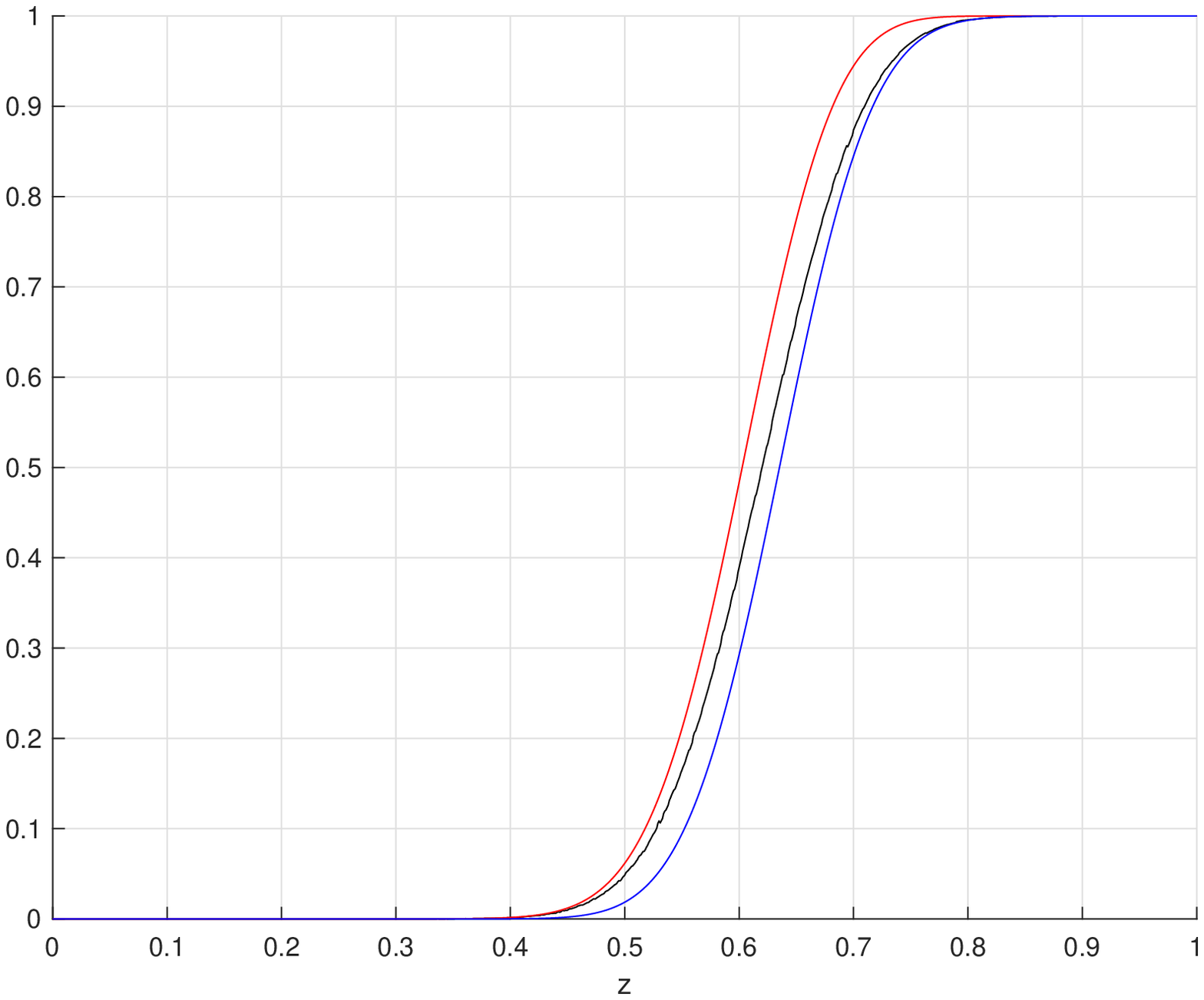}
		\includegraphics[width=8cm]{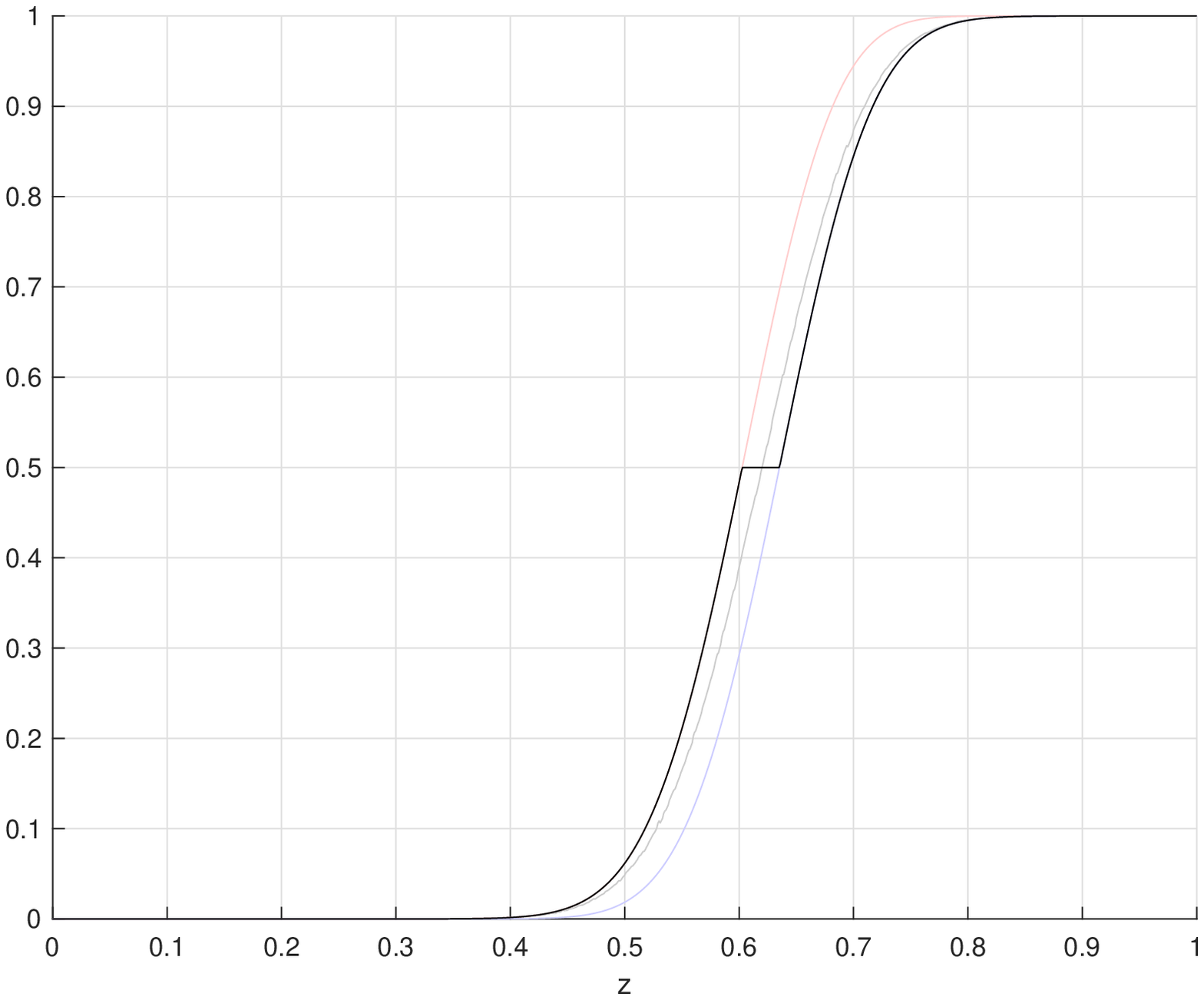}
		\caption{For kernel $K_{64}$, on the left, from left to right we have 
		$\overline{f}_{40}(z)$, $f_{40}(z)$ simulated by Monte Carlo interpolation 
		method and $\underline{f}_{40}(z)$. On the right we have, $\tilde{f}_{40}(z)$.}
		\label{fig:f_40(z)}
	\end{center}
\end{figure}
\noindent For $i=16,17,\cdots,49$ of $K_{64}$. Define
\begin{equation*}
	\overline{E}_{i,w}=\min\left(\sum_{j=1}^i\binom{\ell-i}{j-i}A_i,\binom{\ell}{i}\right),
	\quad\underline{E}_{i,w}=\binom{\ell}{i} - \overline{E}_{\ell+1-i,\ell-w}
\end{equation*}
\noindent And define
\begin{equation*}
    \overline{f}_i(z)=\sum_{i=0}^\ell \overline{E}_{i,w}z^i(1-z)^{\ell-i},\qquad
    \underline{f}_i(z)=\sum_{i=0}^\ell \underline{E}_{i,w}z^i(1-z)^{\ell-i}
\end{equation*}
\noindent Then for $i=16,17,\cdots,49$ and any fixed $z$, 
$f_i(z)\in[\overline{f}_i(z),\underline{f}_i(z)]$. An example for 
$\overline{f}_i(z)$ and $\underline{f}_i(z)$ are shown in Fig~\ref{fig:f_40(z)}.

\begin{figure}[t]
	\begin{center}
		\includegraphics[width=8cm]{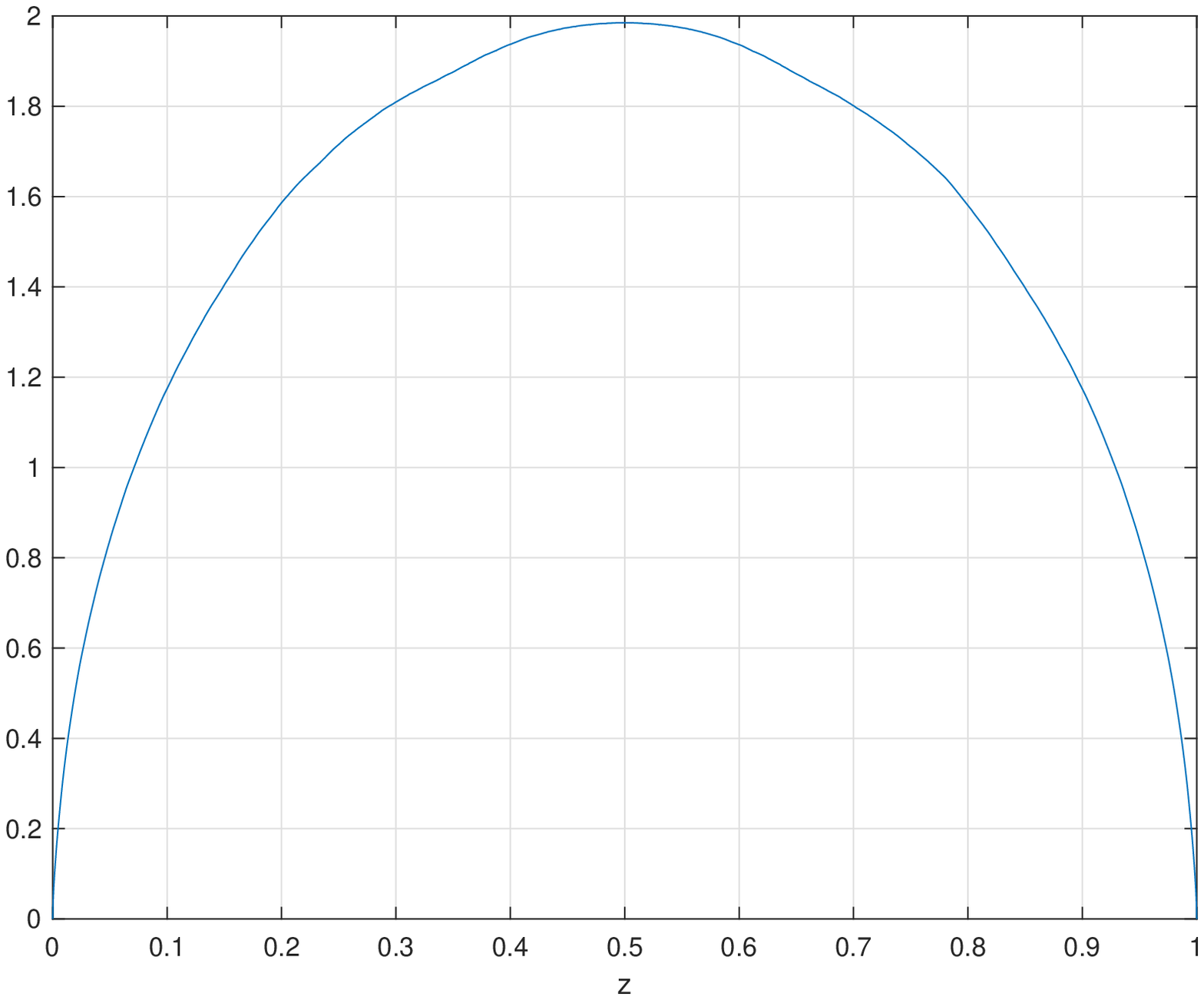}
		\includegraphics[width=8cm]{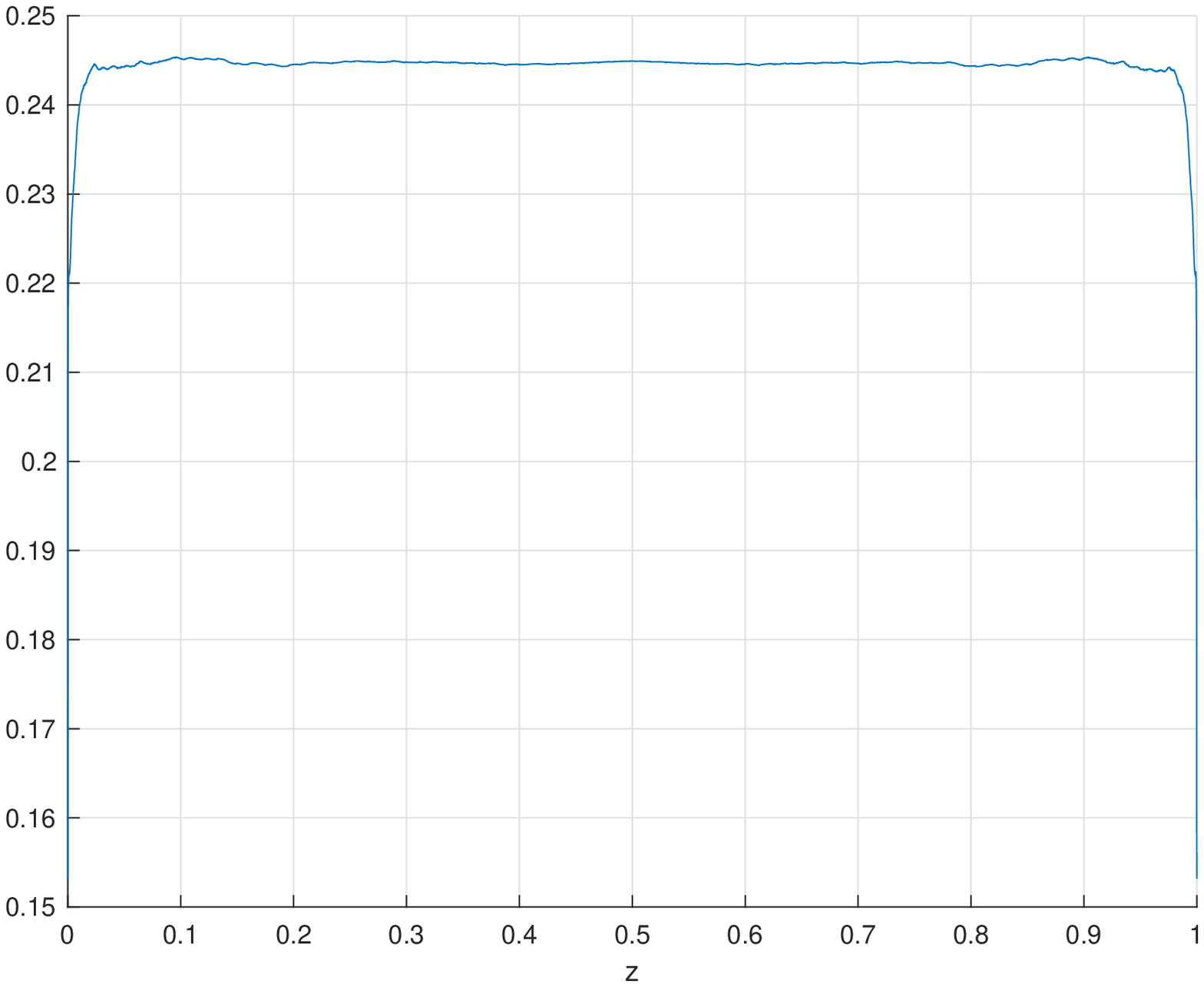}
		\caption{Left: test function $g(z)$. Right: upperbound
		$\overline{b}_0(z)$.}
		\label{fig:g(z)}
	\end{center}
\end{figure}
We define our test function $g$ as follows
\begin{equation*}
	g(z)=\frac{1}{64}\left(\sum_{i=1}^{15}g^\ast(f_i(z))+\sum_{i=50}^{64}g^\ast(f_i(z))+
	\sum_{i=16}^{49} g^\ast(\tilde{f}_i(z))\right),\qquad
	g^\ast(z)=z^{1/2}(1-z)^{1/2}
\end{equation*}
\noindent where
\begin{equation*}
	\tilde{f}_i(z)=\begin{cases}
		\overline{f}_i(z) & \overline{f}_i(z)\le 0.5\\
		0.5 & 0.5\in(\overline{f}_i(z),\underline{f}_i(z))\\
		\underline{f}_i(z) & \underline{f}_i(z)\ge 0.5\\
	\end{cases}
\end{equation*}
An example of $\tilde{f}_i(z)$ is shown in Fig~\ref{fig:f_40(z)}. And a plot of the
test function is shown in Fig~\ref{fig:g(z)}. Since $K_{64}$ is self-dual, by duality 
theorem we can shown that $g(z)$ increases on $[0,0.5]$, decreases on $[0.5,1]$, and
reach its maximum when $z=0.5$. Therefore for $i=16,17,\cdots,49$:
\begin{equation*}
	g(f_i(z))\le \begin{cases}
		g(\overline{f}_i(z)) & \overline{f}_i(z)\le 0.5\\
		g(0.5) & 0.5\in(\overline{f}_i(z),\underline{f}_i(z))\\
		g(\underline{f}_i(z)) & \underline{f}_i(z)\ge 0.5\\
	\end{cases}
\end{equation*}
And this gives us a strict upper bound 
$\overline{g}_1(z)$ for $g_1(z)$:
\begin{align*}
	g_1(z) &= \frac{1}{64}\left(\sum_{i=1}^{64}g(f_i(z))\right)\quad
	\le\quad \overline{g}_1(z)=\frac{1}{64}\left(\sum_{i=1}^{15}g(f_i(z))+\sum_{i=50}^{64}g(f_i(z))+
	\sum_{i=16}^{49} g(\tilde{f}_i(z))\right)
\end{align*}
which provide a strict upper bound $\overline{b}_0(z)$ for $b_0(z)$, as shown in 
Fig~\ref{fig:g(z)}
\begin{equation*}
	b_0(z)=\frac{g_1(z)}{g(z)}\quad\le
	\quad \overline{b}_0(z)=\frac{\overline{g}_1(z)}{g(z)}
\end{equation*}
And the maximum value of $\overline{b}_0(z)$ can be calculated analytically up to any desired
precision. Our calculation shows that:
\begin{equation*}
	b_0=\sup_{z\in(0,1)}b_0(z)\le \max_{z\in(0,1)}\overline{b}_0(z)=0.2454
\end{equation*}
which provides an upperbound $\mu(K_{64})\le -\frac{1}{\log_{64}0.2454}=2.9603$.

\end{appendices}

\end{document}